\documentclass[aps,prx,twocolumn,superscriptaddress,amsmath,amssymb,longbibliography,nofootinbib]{revtex4-2}

\usepackage{amsmath,amsfonts,amssymb,amsthm,bbm,physics}
\usepackage{graphicx}   % need for figures
\usepackage{bbold}
\usepackage{mathtools} % math equal, inclusion arrows

\usepackage{dsfont}
\usepackage[dvipsnames]{xcolor}
\usepackage{colortbl}
\usepackage{cancel}

\usepackage{booktabs}
\usepackage{multirow}
\usepackage{algorithm}
\usepackage{algpseudocode}

\newcommand\dgg{^{\dagger}}
\newcommand\T{\mathsf{T}}
\newcommand\Pf{\mathrm{Pf}}
\renewcommand\det{\mathrm{Det}}
\newcommand\id{\mathbbm{1}}

\newcommand\hilb{\mathcal{H}}

\newcommand{\vol}{\textrm{vol}}

\DeclarePairedDelimiterX{\infdivx}[2]{(}{)}{%
  #1\;\delimsize\|\;#2%
}
\DeclareSymbolFont{Letters} {U}{zeur}{m}{n}% Euler
\newcommand{\KL}{H\infdivx}
\newcommand\av[1]{\left\langle #1 \right\rangle}

\newcommand{\inv}{^{\!\text{-}\hspace{-0.5pt}1}}
\newcommand{\numinv}[1]{^{\text{-}\hspace{-0.5pt}#1}}

\newcommand{\U}{\mathrm{U}}

\newcommand{\SO}{\mathrm{SO}}
\newcommand{\Sp}{\mathrm{Sp}}

\renewcommand{\T}{\textsf{\scriptsize T}}
\newcommand{\R}{\mathbb R}
\newcommand{\C}{\mathbb C}

\newcommand{\x}{\vb x}
\newcommand{\y}{\vb y}
\newcommand{\z}{\vb z}
\newcommand{\s}{\vb s}
\renewcommand{\S}{\mathcal{S}}
\newcommand{\E}{\mathbb{E}}
\newcommand{\Var}{\mathrm{Var}}
\newcommand{\dirich}[1]{\mathcal{E}(#1)}

\newcommand{\pas}{_{\textrm{pas}}}

\definecolor{strcyan}{RGB}{33,140,141}
\definecolor{bluegrey}{RGB}{96,125,139}
\newcommand{\tc}[1]{{\color{magenta} TC[#1]}}

\usepackage{thmtools, thm-restate}
\newtheorem{definition}{Definition}
\newtheorem{corollary}{Corollary}
\newtheorem{theorem}{Theorem}
\newtheorem{lemma}{Lemma}
\newtheorem{proposition}{Proposition}
\newtheorem{remark}{Remark}

\usepackage[breaklinks=true]{hyperref}
\hypersetup{
	colorlinks   = true, %Colours links instead of ugly boxes
	urlcolor     = blue, %Colour for external hyperlinks
	linkcolor    = blue, %Colour of internal links
	citecolor   = magenta %Colour of citations
}

\makeatother

\begin{document}

\renewcommand{\a}{\alpha}
\renewcommand{\b}{\beta}
\newcommand{\I}{\mathcal I} 
\newcommand{\Itil}{\widetilde{\mathcal{I}}}

\newcommand{\vsto}[1][3pt]{\hspace{-1.5pt}\mathrel{%
   \vcenter{\hbox{\rule[-.2pt]{#1}{.4pt}}}%
   \mkern-4mu\hbox{\usefont{U}{lasy}{m}{n}\symbol{41}}}\hspace{-2pt}} % very short \to arrow

% ---------------------------------------------------------------- edit flags
\newif\ifedits
\editstrue                 % <-- \editsfalse hides every annotation
\definecolor{ccblue}{RGB}{0,90,160}
\definecolor{ccmag}{RGB}{190,0,140}
% REVTeX uppercases section titles, which uppercases the colour name too;
% defining the alias lets \tc{} be used inside \section{}.
\definecolor{CCMAG}{RGB}{190,0,140}
\definecolor{ccorr}{RGB}{0,60,210}
\definecolor{CCORR}{RGB}{0,60,210}
\ifedits
  \newcommand{\CC}[1]{{\color{ccblue}\footnotesize\sffamily
    [\textbf{ed:}~#1]}}
  % \tc{...} = content added in the magic/non-Gaussianity round (magenta);
  % \tcn{...} = the accompanying edit note.
  %\newcommand{\tc}[1]{{\color{ccmag}#1}}
  \newcommand{\tcn}[1]{{\color{ccmag}\footnotesize\sffamily
    [\textbf{new:}~#1]}}
  % \bc{...} = content corrected in the referee-comment round (blue);
  % \bcn{...} = the accompanying correction note.
  \newcommand{\bc}[1]{{\color{ccorr}#1}}
  \newcommand{\bcn}[1]{{\color{ccorr}\footnotesize\sffamily
    [\textbf{corr:}~#1]}}
\else
  \newcommand{\CC}[1]{}
  \newcommand{\tc}[1]{#1}
  \newcommand{\tcn}[1]{}
  \newcommand{\bc}[1]{#1}
  \newcommand{\bcn}[1]{}
\fi

% APS double-struck identity (neither bbm, bbold nor dsfont is guaranteed;
% \openone comes with revsymb and is the APS-sanctioned symbol).
%\newcommand{\id}{\openone}

\title{Certifying fermionic Gaussian states (and a little more) with optimal precision dependence}

\author{Ninnat Dangniam}
\email{ninnatdn@gmail.com}
\author{Laphas Premcharoen}
\email{laphas173@gmail.com}
\author{Metrasit Sripech}
%\email{}
\affiliation{
The Institute for Fundamental Study (IF),
Naresuan University,
Phitsanulok 65000, Thailand}

\author{Thiparat Chotibut}
\email{thiparatc@gmail.com}
\affiliation{
Chula Intelligent and Complex Systems Center of Excellence, Department of Physics,
Faculty of Science, Chulalongkorn University, Bangkok 10330, Thailand}
\affiliation{
School of Physical and Mathematical Sciences, Nanyang Technological University, Singapore}

\begin{abstract}
Fermionic Gaussian states are the workhorse reference states for qubit-based quantum simulation of fermionic systems, yet existing certification protocols either proceed via fidelity estimation, leading to suboptimal sample complexity in the target precision $\epsilon$, only apply to Haar-typical states, or require adaptive measurements on a large fraction of qubits.
We give an adaptive protocol that certifies any $d$-mode pure fermionic Gaussian state using $O(d^2\epsilon^{-1})$ copies, single-qubit measurements with only one qubit measured adaptively per copy, and $O(d^3)$ classical processing time per copy.
The $\epsilon^{-1}$ dependence is optimal for the certification problem, 
even among strategies using entangled measurements.
The sample complexity is controlled by the spectral gap of a $d\leftrightarrow d-2$ down-up walk---a Markov chain studied in the theory of high-dimensional expanders, equivalent to a two-site Glauber dynamics in statistical physics---relating certification efficiency to relaxation time to equilibrium.
The worst-case bound is tight for this protocol and is attained by physically relevant states, including ground states of the fully dimerized Su--Schrieffer--Heeger (SSH) chain and BCS pair states.
In contrast, numerics up to $d=14$ suggest that $O(d \epsilon^{-1})$ copies suffice to certify Haar-typical Gaussian states,  a factor of $d$ improvement over the worst-case bound.
%We numerically observe up to $d=14$ that Haar-typical Gaussian states can be certified using $\Theta(d \epsilon^{-1})$ copies, a factor of $d$ improvement over the worst-case bound. 
Finally, because the bound depends only on the target's computational-basis distribution, the protocol extends to \emph{efficiently phase-dressed} states---states obtained by injecting an efficiently computable diagonal phase to Gaussian states---with the same sample complexity. The class includes a continuous family of four-mode non-Gaussian magic states for matchgate computation.
\end{abstract}

\maketitle

\tableofcontents 

%-------------------------------------
\section{Introduction}\label{sec:intro}
%-------------------------------------
 
Every quantum device that claims to have prepared a particular state invites the question of how one would know. The task of \emph{quantum state certification}, testing whether a prepared state is close to a given target state, is fundamental to benchmarking quantum devices, validating the outcomes of experiments, and verifying that the correct quantum computation has been performed \cite{eisert2020quantum,kliesch2021theory}. Dedicated certification tests have been developed across physical platforms, from photonic and bosonic Gaussian states \cite{aolita2015photonic} to fermionic simulations \cite{gluza2018fidelity}.
Unlike full quantum state tomography, whose sample complexity grows exponentially with the number of qubits $d$ \cite{haah2016sample}, certification can be dramatically more resource-efficient. For certain structured families such as stabilizer states, which can be highly entangled, certification can be carried out using only single-qubit measurements and a number of samples independent of $d$ at fixed precision and confidence \cite{flammia2011direct,dasilva2011practical,pallister2018optimal}.

Existing approaches can be organized into three families, distinguished by their underlying goals and assumptions. \emph{Fidelity estimation} \cite{flammia2011direct,dasilva2011practical}, including shadow estimation \cite{huang2020predicting}, quantifies how far a prepared state is from the target state. It is naturally robust to imperfect preparations, but fundamentally it solves a parameter estimation problem rather than certification, requiring $N=\Theta(\epsilon^{-2})$
samples to achieve infidelity $\epsilon$.
\emph{Quantum state verification} (QSV) \cite{pallister2018optimal,zhu2019framework} instead applies a binary test that the target passes with certainty and that any state at infidelity $\epsilon$ fails with non-negligible probability; accepting only on $N$ consecutive passes certifies infidelity $\epsilon$ with $N=O(\epsilon^{-1})$ copies.
The $\epsilon^{-1}$ dependence is optimal for the verification problem, as it matches the scaling attained by simply projecting onto the target, which is optimal even among protocols using entangled measurements.
A third and more recent family, \emph{oracle-based single-qubit certification}, certifies Haar-typical or arbitrary pure states from single-qubit measurements, given classical query access to amplitudes of the target \cite{huang2025certifying,gupta2026few,coladangelo2026robust}. 

Pure fermionic Gaussian states (fGSs)---the ground states of quadratic fermionic Hamiltonians, equivalently the outputs of matchgate or fermionic linear-optical circuits \cite{valiant2002quantum,knill2001fermionic,terhal2002fermion}---are a natural target class for an efficient certificate, being both physically central and classically tractable. 
Slater determinants, the number-conserving subclass of fGSs, are the reference states of variational quantum chemistry and quantum algorithms for the Fermi--Hubbard model \cite{wecker2015solving}, while general fGSs including Bardeen--Cooper--Schrieffer (BCS) states serve the same role for pairing Hamiltonians \cite{jiang2018quantum}. 
Under the Jordan--Wigner transformation, ground states of the transverse-field Ising and XY chains are fGSs, making them standard benchmark states for spin-based simulators. 
fGSs admit preparation by linear-depth Givens-rotation circuits \cite{kivlichan2018quantum,jiang2018quantum} that are now routinely implemented on physical hardware. For example, Google's superconducting-qubit studies of the Fermi--Hubbard model and quantum chemistry operate at the scale of tens of qubits \cite{arute2020observation,google2020hartree},
and cold-atom simulators realize the free-fermion limit of the Hubbard model \cite{tarruell2018quantum}.
The fact that fGSs can be efficiently prepared does not render certification superfluous, but rather more pressing. For instance, the Google Hartree--Fock experiment \cite{google2020hartree} requires substantial error-mitigation in the form of post-hoc correction. Without state certification, an undetected error in the reference state propagates downstream, potentially invalidating the entire simulation run.

Existing approaches to certifying fGSs trade off along multiple axes: sample complexity, measurement adaptivity, classical processing cost, and the class of states to which they apply.
Fidelity witnesses for fermionic simulations \cite{gluza2018fidelity} are robust and use only non-adaptive single-qubit measurements, but are estimation-based and require $O(d^3 \epsilon^{-2})$ samples. 
The shadow-overlap protocol of Huang, Preskill, and Soleimanifar \cite{huang2025certifying} certifies all but an exponentially small fraction of states with $O(d^2\epsilon^{-1})$ samples and single-qubit measurements, but its typicality guarantee does not extend to the structured family of Gaussian states.
The decision-tree protocol of Gupta, He, and O'Donnell \cite{gupta2026few} removes the typicality assumption, being able to certify every pure state with $O(d\epsilon^{-1})$ copies, but requires oracle access to outcome probabilities in arbitrary product bases and measures every qubit adaptively. 
We also note that the task we are considering is different from that of \emph{Gaussianity testing}---deciding whether a prepared state is close to some fGS or far from \emph{every} fGS, a membership testing problem---which has seen significant progress in recent years \cite{lyu2024convolution,coffman2025measuring,bittel2025trace-distance,mele2025few_fermionic,haug2026practical,leone2026fermionic} but concerns a different question.

We ask whether, for \emph{all} fermionic Gaussian states, the optimal $\epsilon^{-1}$ precision scaling, single-qubit measurements, and efficient classical processing can be achieved simultaneously, and answer in the affirmative.

\paragraph*{Results.} 
We give a protocol (Algorithm~\ref{algo}) 
that certifies any $d$-mode pure fermionic Gaussian state of definite parity to infidelity $\epsilon$ with confidence $1-\delta$ using (Theorem~\ref{thm:main})
\begin{align}\label{eq:intro-samples}
    N(\epsilon,\delta,d) \le \left\lceil \frac{d(d-1)}{\epsilon}\ln\frac{1}{\delta} \right\rceil 
\end{align}
copies, single-qubit measurements only, and $O(d^{3})$ classical processing per copy. 
The protocol randomizes between all-Pauli-$Z$ measurements, which certify that the state lies in a sector of fixed fermionic parity, and an adaptive test measuring $d-2$ qubits in the computational basis, one qubit in the $X$ basis, and the final qubit in a basis determined by the preceding outcomes. 

The $d$ dependence solely comes from the spectral gap of the so-called verification operator, for which we prove a lower bound and exhibit an explicit family of Gaussian states that saturates it with matching constant.
Exact numerics up to $d=14$ (Fig.~\ref{fig:gap}) corroborate the bound and indicate that it is pessimistic by a factor of $d$, i.e. $N\approx O(d)$, for Haar-random Gaussian states.

The bound depends on the target state only through its computational-basis distribution. Hence, our protocol in fact certifies with the same efficiency a strictly larger class than the Gaussian states (Corollary~\ref{cor:dressed}) that we call \emph{efficient phase-dressed states} (Definition~\ref{def:dressed}), obtained by adding efficiently computable relative phases to Gaussian states. The class includes, for instance, the continuous family of non-Gaussian states $(\ket{0000}+\ket{0011}+\ket{1100}+e^{i\theta}\ket{1111})/2$, $\theta \neq 0 \pmod{2\pi}$, which are magic states for matchgate computation \cite{hebenstreit2019all}.

The remainder of the paper is organized as follows. 
Sec.~\ref{sec:related} compares our result with prior work. Section~\ref{sec:background} fixes notation and recalls the verification framework and the Gaussian-state formalism for fermions. Sec.~\ref{sec:results} presents the protocol and main results, and Sec.~\ref{sec:conclusion} provides a summary and discusses open problems. Lengthy proofs are deferred to the appendices.

\begin{table*}[t]
	\centering
\setlength{\tabcolsep}{0pt}
\renewcommand{\arraystretch}{1.3}
\begin{tabular}{|>{\hspace{6pt}}c<{\hspace{6pt}}|
    >{\hspace{6pt}}c<{\hspace{6pt}}|
    >{\hspace{6pt}}c<{\hspace{6pt}}|
    >{\hspace{6pt}}c<{\hspace{6pt}}|
    >{\hspace{6pt}}c<{\hspace{6pt}}|}
\hline
\rowcolor{gray!30} 
\textbf{Approach} 
& \textbf{Applied to which (pure) states?}
& \textbf{Adaptive?} 
& \textbf{Oracle?}
& \textbf{Sample complexity} \\
\hline

%===============================
\multirow{2}{*}{Fidelity witness \cite{gluza2018fidelity}} 
& Gaussian 
& \multirow{2}{*}{No} 
& \multirow{2}{*}{No} 
& $O(d^3\epsilon^{-2})${\color{blue}${}^*$} \\

& \cellcolor{gray!30}  Gaussian (gapped Hamiltonian) 
&&
& \cellcolor{gray!30}  $O(d^2 \log^2 (d)\,\epsilon^{-2})$ \\
\hline

%===============================
\multirow{2}{*}{Shadow overlap \cite{huang2025certifying}}
& \multirow{2}{*}{Almost all} 
& No 
& \multirow{2}{*}{Computational basis}  
& $O(d^4\epsilon^{-2})$ \\

&& \cellcolor{gray!30} Single qubit
&
& \cellcolor{gray!30} $O(d^2 {\color{red}\epsilon^{-1}})$ \\
\hline

Decision-tree basis \cite{gupta2026few}
& All 
& $d-1$ qubits 
& All product bases 
& $O(d \,{\color{red}\epsilon^{-1}})$ \\
\hline

\rowcolor{gray!30}
\textbf{QSV (This work)} 
& Gaussian + Phase-dressed 
& Single qubit 
& No
& $O(d^2 {\color{red}\epsilon^{-1}})$ \\
\hline

\end{tabular}
        \caption{Comparison of representative sample-efficient certification protocols using only single-qubit measurements. The table compares protocols based on fermionic fidelity witness \cite{gluza2018fidelity}, shadow overlap \cite{huang2025certifying}, decision-tree basis \cite{gupta2026few}, and our work. 
        The second column specifies the class of pure states for which the protocol provides certification: ``Gaussian" means fermionic Gaussian states, ``gapped Hamiltonian'' indicates Gaussian states that are unique ground states of gapped local Hamiltonians, and ``almost all'' means that the protocol only fails on an $O(2^{-d})$ fraction of pure states. 
        The third column indicates whether the protocol is adaptive and, if so, the number of qubits requiring adaptive measurements. 
        The fourth column indicates whether the protocol requires an amplitude oracle, and if so, in which basis the amplitudes are required.
        The final column summarizes the sample complexity dependence on the number of qubits $d$ and the target infidelity $\epsilon$. All protocols are probabilistic and require $O(\ln\delta^{-1})$ copies to certify the state to fidelity at least $1-\epsilon$ except with failure probability at most $\delta$.
        Red text indicates optimal $\epsilon^{-1}$ scaling, matching that of the best entangled strategy, namely, performing the measurement that directly projects onto the target state. 
        The two shadow-overlap rows are the two variants of Theorem~1 of Ref.~\cite{huang2025certifying}: non-adaptive single-qubit Pauli measurements at $O(\tau^{2}\epsilon^{-2})$, and general (adaptive) single-qubit measurements at $O(\tau\epsilon^{-1})$, with relaxation time $\tau=O(d^{2})$ for almost all pure states.
        }
        ${}^*${This entry is our improved analysis of the witness of Ref.~\cite{gluza2018fidelity}, obtained in Sec.~\ref{sec:related}; the bound stated in Ref.~\cite{gluza2018fidelity} is $O(d^{4}\epsilon^{-2})$.}
		\label{table:comparison}
\end{table*} 

%-------------------------------------
\section{Related works}\label{sec:related}
%-------------------------------------

Our protocol relates to three lines of work: fidelity witnesses for fermionic simulations, the quantum state verification (QSV) framework, and the recent oracle-based certification protocols based on single-qubit measurements. Table~\ref{table:comparison} summarizes the quantitative comparison targeting the class of fermionic Gaussian states. Below we discuss related work in the three families.

\paragraph*{Fidelity witnesses.} The fidelity witness of Gluza \emph{et al.}\ \cite{gluza2018fidelity} estimates the fidelity to a $d$-mode target fGS $\ket{\psi}$ with covariance matrix $\Gamma(\psi)$ (defined in Sec.~\ref{sec:gauss}) to precision $\epsilon$ and confidence $1-\delta$ using
\begin{align}
    N \le \left\lceil \frac{\norm{\Gamma(\psi)}_1^2}{2\epsilon^{2}} \ln\frac{2}{\delta} \right\rceil
\end{align}
samples, where
\begin{align}
    \norm{\Gamma(\psi)}_1 \coloneqq \sum_{1\le j<k\le 2d} \abs{\Gamma_{jk}(\psi)} 
\end{align}
is the entry-wise $\ell_1$ norm of the vector consisting of the strictly upper triangular elements of $\Gamma$. Ref.~\cite{gluza2018fidelity} states their worst-case bound as $N = O(d^4\epsilon^{-2}\ln \delta^{-1})$ by bounding each entry individually: $\abs{\Gamma_{jk}} \le 1$. However, the bound can be improved to $N = O(d^3\epsilon^{-2}\ln \delta^{-1})$ by
applying a standard norm inequality, $\norm{u}_1 \le \sqrt{M} \norm{u}_2$,
for vectors of length $M$. 
Here,
\begin{align}
    \norm{\Gamma}_1^2 \le \binom{2d}{2} \norm{\Gamma}_2^2,
\end{align}
and, since $\Gamma$ is antisymmetric with $\Gamma^{\T}\Gamma = \id_{2d}$ for a pure fGS,
\begin{align}
    \norm{\Gamma}_2^2 = \frac{1}{2} \Tr(\Gamma^{\T}\Gamma) = d.
\end{align}
Consequently, $\norm{\Gamma}_1^2 \le \binom{2d}{2} d = O(d^3)$.
In the most favorable regime of exponentially decaying correlations, as in unique ground states of gapped local quadratic Hamiltonians, their bound improves to $O(d^{2}\log^{2}d\epsilon^{-2})$ \cite{gluza2018fidelity}. Our protocol therefore improves the dimension dependence by a factor of $\log^{2}d$ against their most favorable case, and by $d$ in the worst case. 
Being estimation-based, the fidelity witness is robust to imperfect preparations, but incurs the $\Theta(\epsilon^{-2})$ cost of estimating an expectation value, against our $O(\epsilon^{-1})$. We discuss robustness of general QSV protocols in Sec.~\ref{sec:conclusion}.

\paragraph*{Verification of phased Dicke states.} Within QSV, the work that superficially overlaps with ours is the verification of phased Dicke states \cite{zihao2021PhasedDicke}, which certifies antisymmetric basis states of the form
\begin{align}
    \tfrac{1}{\sqrt{n!}}\sum_{j_1,\dots,j_n}\tilde\epsilon_{j_1,\dots,j_n}\ket{j_1}\otimes\cdots\otimes\ket{j_n},
\end{align}
$\tilde\epsilon$ being the Levi--Civita symbol. These are precisely Slater determinants, Eq.~\eqref{eq:DeterminantProb}, the number-conserving pure fGS. Their protocol shares our adaptive structure, measuring $n-2$ subsystems in a fixed basis to determine the measurement on the remaining two. The measurement models, however, differ fundamentally. Their protocol operates in first quantization, with one verifier per particle acting on a $d$-dimensional single-particle Hilbert space and free to apply any $\mathrm{SU}(d)$ rotation to the measurement basis, whereas ours operates in second quantization, with one verifier per mode restricted to single-qubit basis measurements.
That is, their measurements are particle-local, while our measurements are mode-local, since indistinguishability of particles forbids particle-local measurements. 
Translated into our setting, their particle-local measurements would be entangled across all $d$ modes. The two results therefore certify the same states (in the number-conserving case) under incomparable measurement constraints.

\paragraph*{Oracle-based certification.} The two representative protocols differ in the strength of the oracle they require: Huang, Preskill, and Soleimanifar's shadow overlap protocol \cite{huang2025certifying} queries amplitudes in the computational basis only, whereas Gupta, He, and O'Donnell's decision-tree-basis protocol \cite{gupta2026few} requires probabilities in arbitrary product bases. Our protocol is closest to the shadow overlap protocol, as we leverage the fact that the conditional computational-basis amplitudes are efficiently computable for fGS \cite{valiant2002quantum,knill2001fermionic,terhal2002fermion}. 
In particular, our protocol can be thought of as a variant of the level $m=2$ shadow-overlap protocol performed within a fixed-parity subspace. We adopt the technique of Ref.~\cite{huang2025certifying} that transforms the verification operator into a Markov transition matrix, which we call the \emph{certification walk}, Eq.~\eqref{eq:certification_walk}, with differences designed around the fermionic parity superselection rule: the walk contains no single-site transitions, which would violate parity conservation, and our protocol measures the second-to-last qubit in the $X$ basis rather than $Z$. This last choice is inspired by the decision-tree-basis protocol: because of parity superselection, the $X$ basis (or any equatorial basis on the Bloch sphere) is one in which all conditional Gaussian states are phase states in the sense of Ref.~\cite{gupta2026few}. We adopt this design principle without requiring their full product-basis oracle.

The main difference between our work and Ref.~\cite{huang2025certifying} is that Haar-average-case arguments of the type used in \cite{huang2025certifying} are unavailable for fermionic Gaussian states, since they do not even form a state $2$-design on a fixed-parity sector (a standard computation; see, e.g., \cite{leone2026fermionic}). We instead establish a spectral-gap bound directly via the Pfaffian structure in the computational-basis distribution of Gaussian states.

%-------------------------------------
\section{Notation and background}\label{sec:background}
%-------------------------------------

%-------------------------------------
\subsection{Quantum state verification based on single-qubit measurements}\label{sec:qsv}
%-------------------------------------

Suppose that a laboratory produces a sequence of $d$-qubit states, $\rho_1,\ldots,\rho_N$
that are not necessarily identical, but are assumed to be uncorrelated across preparation rounds so that the joint state is
\begin{align}
\rho_1\otimes\cdots\otimes\rho_N.
\end{align}
We consider the task of certifying whether the lab state is a known pure target state $\ket{\psi}$. Given a user-defined infidelity threshold $\epsilon>0$, the certification protocol distinguishes between the null hypothesis
\begin{align}
H_0:\quad \rho_j=\dyad{\psi},
\end{align}
and the alternative that the prepared states are at least $\epsilon$ away on average from the target in terms of the fidelity. That is, by defining the infidelity of the $j$th copy as $\epsilon_j=1-\expval{\rho_j}{\psi}$, the alternative hypothesis can be expressed as
\begin{align}
H_1:\quad
\bar{\epsilon}:=\frac{1}{N}\sum_{j=1}^N\epsilon_j\geq\epsilon.
\end{align}
At the end of a certification protocol, we either accept or reject $H_0$.
There are two types of possible errors,
\begin{align}
&\Pr[\mathrm{REJECT}\mid \rho_j=\dyad{\psi}\ \forall j] \le \delta',
&&\text{type-I error},\\
&\Pr[\mathrm{ACCEPT}\mid \bar{\epsilon}\geq\epsilon] \le \delta,
&&\text{type-II error}.
\end{align}
A certification protocol is said to be \emph{complete} if $\delta'$ can be made negligible, and likewise \emph{sound} if $\delta$ can be made negligible. 

\paragraph*{Perfect completeness.}
We follow the framework for quantum state verification (QSV) by Pallister \emph{et al.}~\cite{pallister2018optimal}, (see also prior works by Hayashi \emph{et al.}~\cite{hayashi2006LOCC,hayashi2009group}) with a subsequent refinement by Zhu and Hayashi~\cite{zhu2019framework}, in which we impose the stronger requirement of \emph{perfect completeness}: the target state is accepted with certainty so that $\delta'=0$.
Let $\Omega$ denote the positive operator corresponding to the $\mathrm{ACCEPT}$ outcome. Perfect completeness requires
\begin{align}
\Omega\ket{\psi}=\ket{\psi},
\label{eq:verification_operator}
\end{align}
i.e., $\ket{\psi}$ is an eigenstate of $\Omega$ with eigenvalue 1. Such an operator is referred to as a \emph{verification operator}~\cite{zhu2019framework}. If arbitrary collective measurements were available, the optimal choice would simply be the projection onto the target, $\Omega=\dyad{\psi}$. However, implementing such an entangled $d$-qubit measurement may itself demand resources comparable to preparing $\ket{\psi}$ itself. We instead seek to make as few assumptions as possible about the available measurement capabilities.

For several classes of multi-partite states such as stabilizer states \cite{pallister2018optimal,dangniam2020optimal}, hypergraph states \cite{morimae2017hypergraph,zhu2019hypergraph}, Dicke states \cite{liu2019dicke,zihao2021PhasedDicke}, and ground states of frustration-free Hamiltonians \cite{zhu2024frustration-free}, efficient certification can be achieved using only single-qubit (or single-qudit) measurements performed sequentially on individual copies. In this setting, one may regard the $d$ qubits as being accessed by $d$ spatially separated verifiers who are allowed to communicate classically after performing their respective measurements. For a given measurement setting $j$, verifier $k$ performs a binary measurement
\begin{align}
    \left\{M_{k,j}^{(0)},\; M_{k,j}^{(1)}\right\}
    \;\coloneqq\;
    \left\{\dyad{\varphi_{k,j}},\; \id - \dyad{\varphi_{k,j}}\right\}
\end{align}
on qubit $k$. 
The verifiers agree in advance on a set $\mathcal A_j \subseteq \{0,1\}^d$ of outcome strings that count as a pass. That is, the copy passes the test if the joint outcome $\z = (z_1,\dots,z_d)$ lies in $\mathcal A_j$, corresponding to the global pass operator
\begin{align}
    E_j \;=\; \sum_{\z\in \mathcal A_j}\; \bigotimes_{k=1}^{d} M_{k,j}^{(z_k)}.
\end{align}
Generally, $E_j$ is not rank one unless $\mathcal A_j$ is a singleton.
Choosing measurement setting $j$ with probability $p_j$, the resulting verification operator is a convex combination of the individual test operators,
\begin{align}
\Omega=\sum_{j=1}^m p_j E_j,
\qquad
p_j\geq0,\quad\sum_{j=1}^m p_j=1,
\label{eq:Omega_convex}
\end{align}
where $m$ is the number of measurement settings. The certification protocol outputs ACCEPT if and only if all $N$ copies pass their tests.

For an entangled target state, a single measurement setting ($m=1$) is  insufficient to satisfy perfect completeness because the corresponding rank-one product projector $E_1$ cannot have an entangled state as an eigenvector with eigenvalue 1. Multiple measurement settings are therefore required to construct a verification operator whose 1-eigenspace contains the target state.

\paragraph*{Soundness.}
Perfect completeness alone does not guarantee that the protocol can distinguish the target from other states. In particular, any state within the 1-eigenspace of $\Omega$ orthogonal to $\ket{\psi}$ would be able to spoof the certification protocol:
the state is accepted with certainty despite being perfectly distinguishable from the target. Soundness is therefore controlled by the spectral gap of the verification operator,
\begin{align}
\gamma(\Omega)=1-\lambda_2(\Omega),
\label{eq:spectral_gap}
\end{align}
where the eigenvalues are ordered in decreasing order, $1=\lambda_1(\Omega)\geq\lambda_2(\Omega)
\geq\cdots\geq0$.

A nonzero spectral gap guarantees that every state that is not the target has a nonzero probability of being rejected.
More precisely, for a state satisfying $\expval{\rho}{\psi}\leq1-\epsilon$,
the maximum single-copy acceptance probability is
\begin{align}
\max_{\expval{\rho}{\psi}\leq1-\epsilon}\!
\Tr(\Omega\rho)
= (1-\epsilon) + \lambda_2(\Omega) \,\epsilon
= 1-\gamma(\Omega)\epsilon.
\label{eq:plm}
\end{align}
Thus, each test detects an undesirable state with probability at least $\gamma(\Omega)\epsilon$. This detection probability can be amplified arbitrarily close to unity by repeating the test. Importantly, this conclusion does not require the states produced in different rounds to be identical~\cite{zhu2019framework}. 
Since the copies are uncorrelated, the probability that all $N$ tests pass factorizes. Using Eq.~\eqref{eq:plm},
\begin{align}
\Pr[\mathrm{ACCEPT}] &= \prod_{j=1}^N\Tr(\Omega\rho_j)
\leq \prod_{j=1}^N \bigl(1 - \gamma(\Omega) \epsilon_j\bigr) \nonumber\\
&\leq \bigl(1-\gamma(\Omega)\bar{\epsilon}\bigr)^N,
\label{eq:passprob}
\end{align}
where the inequality in the last line follows from the arithmetic-geometric mean inequality. Consequently, if $\bar{\epsilon}\geq\epsilon$, the type-II error is bounded by
\begin{align}
\delta
\leq
\bigl(1-\gamma(\Omega)\epsilon\bigr)^N.
\end{align}
Solving for the critical number of copies gives
\begin{align}
N(\epsilon,\delta,\Omega) &=
\left\lceil \frac{\ln\delta} {\ln\bigl( 1-\gamma(\Omega)\epsilon\bigr)} \right\rceil
\leq
\left\lceil \frac{1}{\gamma(\Omega)\epsilon} \ln\frac{1}{\delta} \right\rceil.
\label{eq:Ntests}
\end{align}
Thus, with probability at least $1-\delta$, the protocol certifies that the average infidelity satisfies $\bar{\epsilon}<\epsilon$.

With no restriction on the measurements, the optimal verification operator is the projective measurement $\Omega=\dyad{\psi}$, for which $\gamma(\Omega)=1$. The corresponding sample complexity,
\begin{align}
N_{\mathrm{opt}}
=
\left\lceil
\frac{\ln\delta}{\ln(1-\epsilon)}
\right\rceil
\leq
\left\lceil
\frac{1}{\epsilon}\ln\frac{1}{\delta}
\right\rceil,
\end{align}
therefore represents the fundamental limit of any state certification protocol. 

\paragraph*{Adaptive verification.}
For target states without a stabilizer-like structure, perfect completeness using fixed local tests can be difficult to achieve. Adaptive verification provides a general approach to overcome this limitation by allowing the measurement setting on each qubit to depend on the outcomes of previous measurements~\cite{zihao2021PhasedDicke,huang2025certifying,gupta2026few,yunting2026universal,zhang2025verification}. 

Here we show that any protocol of the following form has perfect completeness. Suppose that the protocol measures the qubits sequentially, with each single-qubit measurement chosen adaptively based on all previous outcomes. Then, for every measurement branch, the final qubit is measured in a basis containing the conditional state of the target $\ket{\psi}$ given the outcomes obtained on all preceding qubits. We claim that the resulting verification operator $\Omega$ satisfies
$\Omega\ket{\psi}=\ket{\psi}$.

To see this, label a complete sequence of outcomes on the first $d-1$ qubits by $\z=(z_1,\ldots,z_{d-1})$. 
The corresponding projectors may depend on the preceding outcomes, and we denote their product by
\begin{align}
M_{\z} = M_{z_1} \otimes M_{z_2|z_1} \otimes \cdots \otimes M_{z_{d-1}|z_1,\ldots,z_{d-2}},
\end{align}
where each $M_{z_j|z_1,\ldots,z_{k-1}}$ acts only on the $k$th qubit. 
The probability of obtaining the branch $\z$ on the target state is
\begin{align}
p_{\z} = \norm{M_{\z}\ket{\psi}}^2.
\end{align}
For every branch with $p_{\z}>0$, define the corresponding normalized conditional state of the final qubit by
\begin{align}
\ket{\psi_{\z}} = \frac{M_{\z} \ket{\psi}}{\sqrt{p_{\z}}}.
\end{align}
Equivalently, introducing the unnormalized conditional state
\begin{align}
\ket{\smash{\overline{\psi_{\z}}}} = M_{\z}\ket{\psi},
\end{align}
we have $\ket{\smash{\overline{\psi_{\z}}}} = \sqrt{p_{\z}} \ket{\psi_{\z}}$.

By construction, the final measurement associated with branch $\z$ contains the projector $\dyad{\psi_{\z}}$
Consequently, the ACCEPT operator corresponding to this branch is
\begin{align}
E_{\z} = M_{\z}\dgg \bigl( \id\otimes\dyad{\psi_{\z}} \bigr) M_{\z}.
\end{align}
where the identity acts on the first $d-1$ first qubits. 
The full verification operator is the average over the possible adaptive branches,
\begin{align}
\Omega = \sum_{\z} E_{\z},
\end{align}
with branches of zero probability under $\ket{\psi}$ contributing nothing.
Acting on the target state, each branch satisfies
\begin{align}
E_{\z}\ket{\psi} &= M_{\z}\dgg \bigl(
\id \otimes \dyad{\psi_{\z}} \bigr) M_{\z}\ket{\psi} \nonumber\\
&= M_{\z}\dgg \bigl( \id \otimes \dyad{\psi_{\z}} \bigr)
\ket{\smash{\overline{\psi_{\z}}}} \nonumber\\
&= M_{\z}\dgg \ket{\smash{\overline{\psi_{\z}}}} \nonumber\\
&= M_{\z}\dgg M_{\z}\ket{\psi}.
\end{align}
Since the first $d-1$ measurements form a complete orthonormal basis,
\begin{align}
\sum_{\z} M_{\z}\dgg M_{\z} =\id.
\end{align}
Hence, the adaptive protocol attains perfect completeness:
\begin{align}
    \Omega\ket{\psi} = \sum_{\z} E_{\z} \ket{\psi} = \ket{\psi}.
\end{align}

Despite the flexibility, adaptivity generally comes at a cost since implementing such protocols  requires the ability to efficiently compute the post-measurement state conditioned on arbitrary subsets of measurement outcomes in order to determine the next measurement setting. Moreover, spectral analysis of the associated verification operator can become substantially more involved for adaptive schemes.

%-------------------------------------
\subsection{Fermionic Gaussian states}\label{sec:gauss}
%-------------------------------------

We denote the set $\{1,\dots,d\}$ by $[d]$, and the collection of all $n$-element subsets of $[d]$ by $\binom{[d]}{n}$. A bitstring is written in bold as $\x = x_1 \cdots x_d \in \{0,1\}^d$, and can be equivalently identified with a subset $X \in \binom{[d]}{n}$ via $j \in X$ if and only if $x_j = 1$.

To set the stage, we introduce the fermionic Fock space, which is isomorphic to the Hilbert space of $d$ qubits via the Jordan--Wigner transformation. A system of $d$ fermionic modes is described by creation and annihilation operators $a_j\dgg$ and $a_j$, for $j \in [d]$, satisfying the canonical anticommutation relations
$\{a_j,a_k\dgg\} = \delta_{jk} \id$ and $\{a_j,a_k\}=0$.
The Fock basis, or (occupation) number basis, is constructed from the vacuum state $\ket{0_F}$ defined by $a_j \ket{0_F} = 0$ for all $j$, and are given by
\begin{align}
\ket{\x} = (a_1\dgg)^{x_1} \cdots (a_d\dgg)^{x_d} \ket{0_F},
\end{align}
where $x_j \in \{0,1\}$ are also the eigenvalues of the number operators $a_j^\dagger a_j$. 
In particular, states with exactly $n$ occupied modes form the standard basis of $\bigwedge^n \C^d \subset \bigotimes^n \C^d$, where $\C^d$ is the single-particle Hilbert space, and the full Fock space is the direct sum over all particle numbers,
\begin{align}
    \hilb_F (\C^d) = \bigoplus_{n=0}^d \bigwedge^n \C^d.
\end{align}

Two important decompositions of the Fock space arise from the spectra of the total number and the parity operators.  The $n$-particle sector is the eigenspace of the total number operator $\texttt{Num} = \sum_{j=1}^d a\dgg_j a_j$ with eigenvalue $n$. The even and odd parity sectors are respectively the $+1$ and $-1$ eigenspaces of the parity operator $\texttt{Par} = (-1)^{\texttt{Num}}$. By the fermionic parity superselection rule, physical states have support entirely within a single parity sector. Denoting the corresponding parity eigenvalue by $\sigma\in\{+1,-1\}$, we let $\Pi_{\sigma}$ denote the projector onto the $\sigma$-parity sector.

The Jordan--Wigner transformation identifies the Fock space with the Hilbert space $(\C^2)^{\otimes d}$ by representing the fermionic operators in terms of $d$-qubit operators:
\begin{align}
    a_j = \bigotimes_{k=1}^{j-1} Z_k \otimes \ket{0}_j\!\bra{1},
\end{align}
which satisfy the canonical anticommutation relations. Under this mapping, the Fock basis is naturally identified with the computational basis. In particular, measuring Pauli-$Z$ on the $j$th qubit reveals whether mode $j$ is occupied, and both the total number and the parity can be inferred from such measurements, since $\texttt{Num} = \sum_{j=1}^d (\id-Z_j)/2$ and $\texttt{Par} = \bigotimes_{j=1}^d Z_j$.
By contrast, other single-qubit measurements such as equatorial measurements on the Bloch sphere of the form $X \cos \varphi + Y \sin \varphi$ 
do not admit a direct fermionic analogue as they fail to conserve parity, but can become resources inherent to qubit-based protocols such as ours.

To introduce Gaussian states, it is convenient to work with the set of $2d$ Hermitian \emph{Majorana operators},
\begin{align}\label{def:majorana}
    c_j = a_j + a_j\dgg, && c_{d+j} = \frac{a_j-a_j\dgg}{i},
\end{align}
which satisfies the canonical anticommutation relations $\{c_j,c_k\} = 2\delta_{jk}\id$. If we consider the algebra generated by the Majorana operators, a distinguished class of transformations called \emph{Gaussian} or \emph{fermionic linear optical} (FLO) transformations is given by unitaries that preserve both the degree of Majorana monomials and fermionic parity \cite{knill2001fermionic}.
In the context of quantum computing, these transformations are also known as (unitary) \emph{matchgates} \cite{valiant2002quantum,knill2001fermionic,jozsa2008matchgates}.
Such transformations realize a (projective) representation of the rotation group $\SO(2d)$ on the Fock space:
\begin{align}
    \Phi: \SO(2d) &\to \U [\hilb_F (\C^d)] \\
    R &\mapsto e^{-iHt}, 
\end{align}
where $H$ is a quadratic Hamiltonian of the form
\begin{align}
    H = \frac{i}{4} \sum_{j,k=1}^{2d} h_{jk} c_j c_k,
\end{align}
with $h$ real antisymmetric, and $R = e^{h} \in \SO(2d)$. 
Since $H$ is quadratic, the unitary $e^{-iHt}$ is an even operator (containing only Majorana monomials of even degree) and commutes with $\texttt{Par} = \prod_{j=1}^d (-i c_j c_{d+j})$.
The conjugation action of $\Phi(R)$ on linear Majorana operators defines a genuine (non-projective) representation, given by
\begin{align}\label{eq:gaussian_unitary_action}
    \Phi(R)\dgg c_j \Phi(R) = \sum_{k=1}^{2d} R_{jk} c_k,
\end{align}
for $R \in \SO(2d)$. 

A pure \emph{fermionic Gaussian state} (fGS)---or Gaussian states in short---is defined as any state obtained from a number state by a Gaussian unitary, i.e., $\ket{\psi} = \Phi(R)\ket{\x}$. 
Such states admit a polynomial-size description in terms of their antisymmetric $2d \times 2d$ \emph{covariance matrix}
\begin{align}
\Gamma_{jk}(\psi) = -\frac{i}{2}\bra{\psi}[c_j, c_k]\ket{\psi},
\end{align}
with $\Gamma^\T \Gamma = \id$, i.e. $\Gamma \in \SO(2d)$, and $\Gamma^{\T}=-\Gamma$, if and only if $\ket{\psi}$ is pure Gaussian.
All two- and three-qubit pure states with definite parity are Gaussian \cite{bravyi2005classical}.
In contrast to bosonic states, every fermionic number state is Gaussian.
Given a bitstring $\x$, consider a $\pm1$ vector $\tilde{\vb x}$ with entries $\tilde{x}_j = (-1)^{x_j}$, and let $D_{\x} = \textrm{diag}(\tilde{\vb x}) \oplus \textrm{diag}(\tilde{\vb x})$. The covariance matrix of the corresponding  number state $\ket{\vb x}$ is
\begin{align}\label{eq:symplectic_form}
    \mqty(0 & \textrm{diag}(\tilde{\vb x}) \\ -\textrm{diag}(\tilde{\vb x}) & 0) = D_{\x}\, J, \qquad
    J = \mqty(0&\id_d\\ -\id_d &0),
\end{align}
where $J$ is the covariance matrix of the vacuum state $\ket{0_F}$ and also the standard symplectic form.
Using Eq.~\eqref{eq:gaussian_unitary_action}, the covariance matrix transforms by conjugation under a Gaussian unitary:
\begin{align}\label{eq:FLO-conjugation-action}
\Gamma \!\bigl( \Phi(R) \dyad{\psi} \Phi(R)\dgg\bigr)
= R\, \Gamma(\psi)\, R^{\mathsf T}.
\end{align}
Thus, the action of a Gaussian unitary can be simulated classically by updating the covariance matrix in $O(d^3)$ time~\cite{valiant2002quantum,terhal2002fermion,knill2001fermionic}.

Number-basis measurements can likewise be simulated efficiently. In particular, the probability of either outcome of a single-mode number measurement can be evaluated in $O(1)$ time, while the covariance matrix of the corresponding post-measurement Gaussian state can be updated in $O(d^2)$ time~\cite{bravyi2012disorder}. 
More generally, the full number-basis distribution of a pure fGS admits a closed-form expression,
\begin{align}\label{eq:PfaffianProb}
\pi(\vb x) \coloneqq \abs{\av{\vb x|\psi}}^2
= \frac{1}{2^d}\, \abs{\Pf \bigl(\Gamma + D_{\x}J\bigr)},
\end{align}
where $\Pf$ denotes the Pfaffian, which can be computed in $O(d^3)$ time. Eq.~\eqref{eq:PfaffianProb}, an important ingredient in the proof for our main theorem, follows from the Pfaffian inner-product formula for Gaussian states, which we  discuss further, along with other  properties of the Pfaffian, in Appendix~\ref{app:grassmann}.

An important subclass of FLO transformations is given by the \emph{passive} or \emph{number-conserving} FLO transformations, generated by quadratic Hamiltonians of the form
$\sum_{j,k=1}^d A_{jk} a\dgg_j a_k$ with $A$ Hermitian. 
They commute with the total number operator, preserving each fixed particle-number sector, and constitute the largest subgroup that leaves the vacuum $\ket{0_F}$ invariant (up to a phase).
In particular, the conjugation action of the corresponding orthogonal transformation $R$ preserves the vacuum covariance matrix: $R J R^{\T} = J$,
so they can be characterized as the subgroup $\SO(2d) \cap \Sp(2d,\R)$ naturally isomorphic to $\U(d)$ thanks to the two-out-of-three property. 
In the Majorana representation, this identification is made explicit by the embedding 
\begin{align}
    R = \mqty(\textrm{Re}\,U & -\textrm{Im}\,U \\ \textrm{Im}\,U & \textrm{Re}\,U) \quad \in \SO(2d).
\end{align}
Correspondingly, the conjugation action of $\Phi(R) = \Phi_{\pas}(U)$ does not mix creation and annihilation operators:
\begin{align}
    \Phi_{\pas}(U)\dgg a_j \Phi_{\pas}(U) = \sum_{k=1}^d U_{jk} a_k,
\end{align} 
making it clear that passive FLO transformations arise as the natural lifting of single-particle unitaries on $\C^d$ to the $n$-particle sector:
\begin{align}
    \Phi{\pas}: \U(d) &\to \U \left(\bigwedge^n \C^d\right) \\
    U &\mapsto U^{\otimes n}\bigr|_{\bigwedge^n \C^d}.
\end{align}

Gaussian states that are also eigenstates of the total number operator, $\texttt{Num}\ket{\psi}=n\ket{\psi}$, 
are known as \emph{Slater determinants}, and are obtained by applying passive FLO transformations to number states, $\ket{\psi} = \Phi_{\pas}(U)\ket{\x}$, which can equivalently be written as
\begin{align}\label{eq:slater}
    \ket{\psi} = \ket{e_1} \wedge \cdots \wedge \ket{e_n},
\end{align}
where $\{\ket{e}_j\}_{j=1}^n$ is an orthonormal set in $\C^d$. 
The overlap between two such states takes the determinantal form \cite{terhal2002fermion},
\begin{align}
    \mel{\y}{\Phi_{\pas}(U)}{\x} = \det(U_{X,Y}),
\end{align}
where $U_{X,Y}$ denotes the submatrix of $U$ obtained by keeping only the rows $j$ for which $x_j=1$ and the columns $k$ for which $y_k=1$. It follows from the Cauchy--Binet formula that the computational-basis probability also reduces to the determinantal form,
\begin{align}\label{eq:DeterminantProb}
    \pi(\x) = \abs{\av{\x|\psi}}^2 = \det(K_{X,X}),
\end{align}
where $K_{X,X}$ is the principal submatrix of $K \coloneqq U_{[d],X} (U_{[d],X})\dgg$, which is also the one-particle reduced density matrix (1-RDM) given by
\begin{align}
    K_{jk} = \expval{a_j\dgg a_k}{\psi}.
\end{align}
with $K$ being a rank-$n$ projection operator.
Since  $\expval{a_ja_k}{\psi}=0$ for any Slater determinant, $K$ contains the same information as the covariance matrix 
via (see, e.g., \cite{christensen2026FLO_learning})
\begin{align}
    \Gamma = \mqty(2\,\textrm{Im}(K)&\id_d - 2 \,\textrm{Re}(K) \\ -\id_d +2\,\textrm{Re}(K) & 2\,\textrm{Im}(K)).
\end{align}
The purity condition for fGS, $\Gamma^{\T}\Gamma=\id$, follows from the fact that $K$ is a projection operator.

%-------------------------------------
\section{Results}\label{sec:results}
%-------------------------------------

\subsection{Protocol}

The certification protocol for a pure fermionic Gaussian state $\ket{\psi}$ proceeds as follows. In each measurement round, consuming one copy of the laboratory state $\rho$, we uniformly choose between two tests: the parity test, implemented by measuring $Z^{\otimes d}=\texttt{Par}$, and the \emph{correlation test}. The full procedure is given in Algorithm~\ref{algo}.
If the outcome string $(\z^{(k)},x)$ has zero probability under the model state $\ket{\psi}$, the test immediately outputs $\mathrm{REJECT}$.

\begin{algorithm}[H]
\caption{Certify fGS $\ket{\psi}$}\label{algo}
\begin{algorithmic}[1]
\Require One copy of an unknown lab state $\rho$, and the covariance matrix $\Gamma$ of the target fGS $\ket{\psi}$
\Ensure $\mathrm{ACCEPT}$ or $\mathrm{REJECT}$
\State \textbf{Flip} a fair coin; if heads, run the \emph{parity test}, else the \emph{correlation test}.
\Statex \textit{Parity test:}
\State \textbf{Measure} all $d$ qubits of $\rho$ in the computational basis and \textbf{ACCEPT} iff the outcome parity matches that of $\ket{\psi}$.
\Statex \textit{Correlation test:}
\State \textbf{Sample} $k=\{k_1,k_2\}\subset[d]$ uniformly at random.
\State \textbf{Measure} the qubits in $[d]\setminus k$ of $\rho$ in the computational basis, obtaining $\z^{(k)}\in\{0,1\}^{d-2}$.
\State \textbf{Measure} qubit $k_1$ of $\rho$ in the $X$ basis, obtaining outcome $x\in\{+,-\}$.
\State Let $\ket{\psi'}$ be the post-measurement state of $\ket{\psi}$ on qubit $k_2$ conditioned on $(\z^{(k)},x)$, computed from $\Gamma$. \textbf{Measure} qubit $k_2$ of $\rho$ in a basis containing $\ket{\psi'}$ and \textbf{ACCEPT} iff the outcome is $\ket{\psi'}$.
\end{algorithmic}
\end{algorithm}

The conditional state required in the last step can be computed in polynomial time. After measuring the qubits in $[d]\setminus k$, the conditional state remains Gaussian and can be obtained efficiently. 
In particular, the resulting two-qubit state is restricted by the parity superselection rule to the form
\begin{align}\label{eq:conditional-two-qubit}
\ket{\psi_{\z^{(k)}}}
= \alpha\ket{b_1b_2}
+ \beta\ket{\smash{\overline b_1\overline b_2}},
\end{align}
where $b_1=0$, $b_2\in\{0,1\}$ and $\overline b_j=b_j\oplus1$. The coefficients $\alpha$ and $\beta$ can be read off from the $4\times4$ conditional covariance matrix. 
Throughout we fix the convention $k_1:=\min k$ and $k_2:=\max k$, so that the qubit measured in the $X$ basis always precedes the adaptively measured qubit in the Jordan--Wigner order.
In our ordering of the Majorana operators, Eq.~\eqref{def:majorana},
\begin{align}
|\alpha|^2-|\beta|^2=\Gamma_{13},
\qquad
\alpha\beta^*= s_{\textrm{JW}}\frac{(-1)^{b_2} \Gamma_{14}+i\Gamma_{12}}{2}.
\end{align}
where
\begin{align}\label{eq:jw-sign}
s_{\mathrm{JW}}=(-1)^{\sum_{k_1<j<k_2} z_j}
\end{align}
is the Jordan--Wigner string picked up by the products $c_{k_1}c_{k_2}$ and $c_{k_1}c_{d+k_2}$ on the measured branch $\z^{(k)}$. For adjacent pairs $k_2=k_1+1$ the product is empty and $s_{\mathrm{JW}}=1$.
Conditioning further on the outcome $x=\pm$ yields the single-qubit state
\begin{align}
\ket{\psi'}=\alpha\ket{b_2}\pm\beta\ket{\smash{\overline b_2}}.
\end{align}
Hence, the state $\ket{\psi'}$ and the corresponding measurement basis in the last step can be computed in $O(1)$ time once the conditional covariance matrix is available.

\subsection{Efficiency bound}

\paragraph*{The verification operator and completeness.} 
Denote by $\sigma$ the parity of the target fGS. The test operator corresponding to the parity test is given by the projection $\Pi_{\sigma}$ onto the $\sigma$-parity sector. Let $E_k$ be the test operator implemented by steps 4-6 of Algorithm~\ref{algo} for fixed subset $k$, 
\begin{align}
    E_k &= \sum_{\z^{(k)}} \dyad{\smash{\z^{(k)}}} \otimes \! \sum_{x\in\{+,-\}} \! \dyad{x} \otimes \dyad{\smash{\psi'_{(\z^{(k)},x)}}} 
\end{align}
and $\Omega_{0}=\binom{d}{2}^{-1}\sum_k E_k$ for its average. $\Omega_0$ is the test operator for the correlation test. The total verification operator corresponding to the protocol is given by
\begin{align}
    \Omega = \frac{\Pi_{\sigma} + \Omega_0}{2}.
\end{align}
Since $\Pi_{\sigma}\ket{\psi}=\ket{\psi}$, and the correlation test has the structure of an adaptive protocol described in Section~\ref{sec:qsv}, the overall protocol achieves perfect completeness.

Now we analyze properties of the protocol. First, one can verify that $\Omega_0$ is block diagonal with respect to the two parity sectors by direct computation. The operator $E_k$ for each $\z^{(k)}$ restricted to the last two qubits is block diagonal, 
\begin{align}\label{eq:condition-verify-matrix}
\begin{array}{c@{\;}c@{\;}c}
&
\displaystyle
b_1b_2\quad
\overline{b}_1\overline{b}_2 \quad
b_1\overline{b}_2\quad
\overline{b}_1b_2\quad
&
\\[-2pt]
&
\mqty(
|\alpha|^2 & \alpha\beta^* & 0 & 0 \\
\alpha^*\beta & |\beta|^2 & 0 & 0 \\
0 & 0 & |\beta|^2 & \alpha^*\beta \\
0 & 0 & \alpha\beta^* & |\alpha|^2
)
&
\begin{array}{c}
b_1b_2\\
\overline{b}_1\overline{b}_2\\
b_1\overline{b}_2\\
\overline{b}_1b_2\\
\end{array}
\end{array}
\end{align}
with the upper left block (the $\sigma$-parity sector) being precisely the one-dimensional projector onto the conditional two-qubit state $\ket{\psi_{\z^{(k)}}}$ in Eq.~\eqref{eq:conditional-two-qubit}.
That is, the two-dimensional accepted subspace of $E_{\z}$ is spanned by states of
\emph{different} parity.
Thus, inside the parity
sector the two single-qubit measurements of steps 5-6 realise exactly the
rank-one projector onto the two-qubit conditional target even though
$E_k$ itself has rank two on each two-qubit block:
\begin{align}\label{eq:rank1}
  \Pi_{\sigma} E_k \Pi_{\sigma} = \sum_{\z^{(k)}} \dyad{\smash{\z^{(k)}}} \otimes \dyad{\psi_{\z^{(k)}}}.
\end{align}
Define the parity-projected operator
\begin{align}\label{def:L-matrix}
  L:=\Pi_{\pm} \Omega_{\sigma} \Pi_{\sigma}
   =\binom{d}{2}^{\!\!-1} \mkern-9mu \sum_{k}\sum_{\z^{(k)}} \dyad{\smash{\z^{(k)}}} \otimes \dyad{\psi_{\z^{(k)}}}.
\end{align}
Its spectral gap $\gamma(L)$ is what the sample complexity depends on as the next lemma shows.

\begin{lemma}\label{lem:combined}
Let $\Omega = (\Pi_{\sigma} + \Omega_0)/2$ be the verification operator. Then 
\begin{equation}
  \gamma(\Omega) \ge \min \left\{ \frac{\gamma(L)}{2}, \frac{1}{2} \right\}.
  \label{eq:combined}
\end{equation}
\end{lemma}

\begin{proof}
Since both terms are block diagonal in parity, we may bound eigenvalues sector by sector. In the target's sector $\Pi_{\sigma}$ acts as the identity and $\Omega_{0}$ as $L$, so on the orthogonal complement of $\ket{\psi}$ within that sector the eigenvalues are at most $1-\gamma(L)/2$. In the opposite sector $\Pi_{\sigma}$ acts as zero and $\Omega_{0}\le\id$, so the eigenvalues are at most $1/2$. Hence,
$\gamma(\Omega)\ge\min\{\gamma(L)/2,1/2\}$.
\end{proof}

\paragraph*{Soundness.}
To establish soundness, it suffices to show that the spectral gap $\gamma(L)$ does not decay faster than inverse-polynomially in the number of qubits for every target fGS.
Denote by $V = \{\x \mid \pi(\x) >0\}$
the support of the computational-basis distribution $\pi$ of Eq.~\eqref{eq:PfaffianProb}, and consider a graph $G=(V,E)$ where $(\x,\y)\in E$ if and only if they have Hamming distance two, i.e., $|\x \oplus \y|=2$.
As we show in
Appendix~\ref{app:verify-op-is-markov}, $L$ has the same spectrum as the transition matrix $P$ of a Markov chain on $G$ that we call the \emph{certification walk},
\begin{align}\label{eq:certification_walk}
\mel{\x}{P}{\y} = 
    \begin{cases}
        {\displaystyle {d \choose 2}^{\!\!-1} \frac{ \pi(\y) }{\pi(\x)+\pi(\y)}}, 
        & (\x,\y)\in E,  \\
        {\displaystyle
        {d \choose 2}^{\!\!-1} \mkern-9mu \sum_{\substack{\x' \\ |\x \oplus \x'|=2}} \mkern-9mu \frac{\pi(\x)}{\pi(\x) + \pi(\x')} }, 
        & \x=\y, \\
        0, & \textrm{otherwise},
    \end{cases}
\end{align}
having $\pi$ as its stationary distribution. If $V$ has only one vertex, $L=\dyad{\psi}$ and $\gamma(L)=1$.
Our main technical result lower bounds the spectral gap of this walk.

\begin{restatable}[Spectral gap of the certification walk]{theorem}{SpectralGap}
\label{thm:main}
Let $\ket{\psi}$ be an fGS with definite parity and let $L$ be as in \eqref{def:L-matrix}. The spectral gap is lower bounded as
\begin{align}\label{eq:gapbound}
  \gamma(L) \ge \frac{2}{d(d-1)}.
\end{align}
Consequently, the sample complexity of the certification protocol satisfies
\begin{align}
    N(\epsilon,\delta,d) \le \left\lceil \frac{d(d-1)}{\epsilon} \ln \delta^{-1} \right\rceil.
\end{align}
\end{restatable}
\noindent The proof of Theorem~\ref{thm:main}, given in Appendix~\ref{app:main-proof}, proceeds by  identifying the certification walk with a two-site Glauber dynamics. Equivalently, in the language of theoretical computer science, this is a \emph{multi-level down-up walk}, a class of Markov chains that has been extensively studied in the theory of high-dimensional expanders \cite{kaufman2020high_order,alev2020improved,anari2021hardcore,dikstein2024boolean,anari2019logII,anari2022entropic}.
Since the spectral gap is the reciprocal of the relaxation time, the desired bound can be established by showing rapid relaxation of this walk.

What allows us to bound the relaxation time is a chain of existing results that relates the relaxation time to locations of complex zeros of the multivariate polynomial associated with the computational-basis distribution,
\begin{align}
    Z_{\pi}(z_1,\dots,z_d) = \sum_{\x} \pi(\x) \prod_j z_j^{x_j}.
\end{align}
 In particular, the fact that $\Gamma^{\T}\Gamma=\id$ for the covariance matrix of a pure fGS implies that $Z_{\pi}$ is zero-free in the open right-half plane $\{z_j \mid \textrm{Re}(z_j) > 0\}$ for all $j$ (Proposition~\ref{prop:CayleyTransform} and Lemma~\ref{lem:ZeroFreeness}).
 This zero-free property, as it turns out, implies rapid decay of the distinguishability between any distribution and the stationary distribution $\pi$ of the walk as measured by the relative entropy (Kullback--Leibler divergence). The rate of entropy contraction then lower bounds the spectral gap directly (Proposition~\ref{prop:entropy_contraction_bounds_gap}), ensuring fast relaxation. 
 The proof makes use of the Grassmann--Berezin integral representation of the Pfaffian, together with various standard and specialized techniques from the Markov-chain literature. We review the former in Appendix~\ref{app:grassmann} and the latter in Appendix~\ref{app:markov} and \ref{app:downUp}.

A simple family of $d$-mode Gaussian states saturates the spectral-gap bound of Theorem~\ref{thm:main}, showing that the bound is tight.

\begin{proposition}\label{prop:tight}
Consider $m\geq1$ separate pairs of modes, with the other
$d-2m\geq0$ modes fixed in a number state $\ket{\x}$:
\begin{align}\label{eq:pair-example}
\ket{\Psi}=\bigotimes_{r=1}^{m}\ket{\chi_r}\otimes\ket{\x},
\end{align}
where each pair is either
$\ket{\chi_r}=\alpha_r\ket{10}+\beta_r\ket{01}$ or
$\ket{\chi_r}=\alpha_r\ket{00}+\beta_r\ket{11}$, with
$|\alpha_r|^2+|\beta_r|^2=1$ and $\alpha_r\beta_r\neq0$.
The certification walk of $\ket{\Psi}$ has spectrum
\begin{equation}
\left\{ 1 - \frac{s}{c} \: \middle| \: s = 0, 1, \dots, m \right\}, \qquad c = \binom{d}{2},
\end{equation}
the eigenvalue $1 - s/c$ having multiplicity $\binom{m}{s}$, independent of $\alpha_r$ and $\beta_r$. In particular, the spectral gap is
\begin{equation}
\gamma(L) = \frac{1}{c} = \frac{2}{d(d-1)}.
\end{equation}
\end{proposition}
\begin{proof}
Let pair $r$ occupy modes $2r-1,2r$. Each pair state has fixed parity, so the product with the number state $\ket{\x}$ is again a pure Gaussian state of definite parity. Measuring occupations therefore yields, independently for each pair, one of two outcomes, and $2^m$ outcomes in total.

Each step of the certification walk selects two modes uniformly at random from the $c=\binom{d}{2}$ pairs of modes.

\textbf{Case I.} The two selected modes belong to the same pair $r$. The occupation pattern of pair $r$ is then resampled with probabilities $|\alpha_r|^2$ and $|\beta_r|^2$, giving the conditional update
\begin{align}
    T_r  
    &= \id_2^{\otimes(r-1)} \otimes 
    \mqty(
    |\alpha_r|^2 & |\beta_r|^2  \\
    |\alpha_r|^2  & |\beta_r|^2)
    \otimes \id_2^{\otimes(m-r)},
\end{align}
which has eigenvalues $1$ and $0$.

\textbf{Case II.} The two selected modes belong to different pairs. There are $c-m$ such choices, and in this case all occupations are left unchanged.

The transition matrix is thus
\begin{align}
P &= \frac{1}{c} \left( \sum_{r=1}^m T_r + (c-m) \id \right) \nonumber \\
&= \id - \frac{1}{c} \sum_{r=1}^m \bigl( \id - T_r \bigr).
\end{align}
Since each term $\id - T_r$ acts non-trivially only on the $r$th tensor factor, the eigenvalues of the sum are sums of the eigenvalues of the individual terms. Hence $P$ has eigenvalues $1-s/c$, each with multiplicity $\binom{m}{s}$, for $s=0,\ldots,m$, the number of tensor factors $\mqty(1&1)^{\!\T}$ in the right-eigenvector. 
\end{proof}

%------------------------------------
\subsection{Examples and numerical results}\label{sec:numerics}
%------------------------------------

\begin{figure}[t]
\includegraphics[width=\columnwidth]{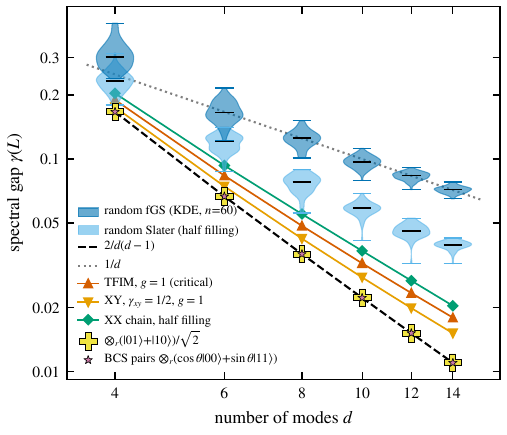}
\caption{Spectral gap $\gamma(L)$ of the certification walk versus the number
of modes $d$, on logarithmic axes. Violin plots show the distribution over
$n=60$ random instances per $d$---Haar-random pure fGS (dark) and random
half-filled Slater determinants (light)---with horizontal lines at the
medians and at the extreme values. Solid curves are the ground states of the
critical transverse-field Ising chain ($g=1$), the XY chain
($\gamma_{xy}=1/2$, $g=1$) and the half-filled XX chain. The two product
families $\bigotimes_r(\cos\theta\ket{00}+\sin\theta\ket{11})$ (stars) and
$\bigotimes_r(\ket{01}+\ket{10})/\sqrt2$ (plus signs) saturate the bound
\eqref{eq:gapbound} of Theorem~\ref{thm:main} (dashed) at every $d$; the
dotted line is $1/d$.}
\label{fig:gap}
\end{figure}

We discuss two worst-case families of Proposition~\ref{prop:tight} in detail since they are not merely extremal examples but correspond to states of direct physical relevance: ground states of the fully dimerized Su--Schrieffer--Heeger (SSH) chain, a canonical model of one-dimensional topological insulators, and BCS states, expressed in a paired-mode basis and used in mean-field descriptions of superconductivity.

\paragraph*{One fermion per pair.}
For $d=2m$, choose $\ket{\chi_r}=\alpha_r\ket{10}+\beta_r\ket{01}$ on every pair. Each pair contains one fermion in a superposition of
its two modes, giving the $m$-particle Slater determinant
\begin{align}\label{eq:slater-dimers}
    \ket{\Psi_{\mathrm{dim}}}
    =\prod_{r=1}^{m}(\alpha_r a_{2r-1}^{\dagger} +\beta_r a_{2r}^{\dagger})\ket{0_F}.
\end{align}
For $\alpha_r=\beta_r=-1/\sqrt2$, the fermion in each pair is equally shared between its two sites. This is the ground state of the SSH chain with one fermionic mode per site, zero hopping between different pairs, and negative hopping coefficients within each pair~\cite{su1979solitons,sirker2014boundary}. Both ends
of the chain belong to complete pairs, so no end site is left unpaired.

The same balanced state can also be the ground state of a connected system~\cite{crepin2011phase}. For $m\geq2$, label the two sites of pair $r$ by $A,B$. Connect the $A$ sites into one chain and the $B$ sites into another; each pair then forms one rung of a ladder. With open ends and real hopping amplitudes, the Hamiltonian is
\begin{align}\label{eq:ladder-parent}
H_{\mathrm{lad}}={}&-t_\perp\sum_{r=1}^{m}
 (a_{r,A}^{\dagger}a_{r,B}+\mathrm{h.c.})\nonumber\\
&-t_\parallel\sum_{r=1}^{m-1}\sum_{\ell=A,B}
 (a_{r,\ell}^{\dagger}a_{r+1,\ell}+\mathrm{h.c.}).
\end{align}
Here $t_\perp$ controls hopping within a pair and $t_\parallel$ controls hopping between neighboring pairs along either chain. The one-particle energies are
$\pm t_\perp-2t_\parallel\cos[j\pi/(m+1)]$, $j=1,\ldots,m$.
For $t_\parallel\neq0$ and $t_\perp>2|t_\parallel|\cos[\pi/(m+1)]$, the $m$ lowest one-particle states span the same space as the symmetric rung states $(\ket{r,A}+\ket{r,B})/\sqrt2$, where $\ket{r,\ell}$ denotes one fermion at site $(r,\ell)$. In particular,
$t_\perp>2|t_\parallel|>0$ is sufficient. With $m$ fermions, all of these lowest states are occupied. Changing the basis of this occupied space changes the Slater determinant only by an overall phase. The ground state therefore remains the balanced state~\eqref{eq:slater-dimers}, and the bound is attained exactly even when neighboring pairs are connected by nonzero hopping.

\paragraph*{Pairs with zero or two fermions (BCS states).}
For $d=2m$, choose $\ket{\chi_r}=u_r\ket{00}+v_r\ket{11}$ on every pair instead. The full state is  
\begin{align}\label{eq:bcs-dimers}
    \ket{\Psi_{\mathrm{BCS}}}
    =\prod_{r=1}^{m}(u_r+v_r a_{2r-1}^{\dagger}a_{2r}^{\dagger})\ket{0_F},
\end{align}
with $|u_r|^2+|v_r|^2=1$ and $u_rv_r\neq0$. Each pair is in a superposition of being empty and containing two fermions. This is the pair structure of the BCS wave function used to describe superconductivity~\cite{bardeen1957theory}. Such states can be prepared by gates acting separately on each pair of qubits and are used as starting states in fermionic simulation algorithms~\cite{jiang2018quantum}.

Figure~\ref{fig:gap} shows the spectral gap $\gamma(L)$ as a function of the number of modes $d$ for $d \le 14$, computed exactly for a range of fermionic Gaussian states. Haar-random fGSs exhibit a gap scaling as $1/d$, 
suggesting that generic Gaussian states are a factor of $d$ less costly to certify than the worst case would guarantee. 
Several physically motivated families besides the worst-case ones are also shown: the critical transverse-field Ising model (TFIM), the XY and XX chains under the Jordan--Wigner transformation.
Details of the numerical simulations are given in Appendix~\ref{app:numerics}.

%-------------------------------------
\subsection{Certifying phase-dressed  states}\label{sec:dressed}
%-------------------------------------

The sample complexity of Theorem~\ref{thm:main} is governed by the spectral gap of the certification walk, which depends on the target through its computational-basis distribution $\pi(\x)=\abs{\av{\x|\psi}}^2$ alone. The protocol therefore extends to all \emph{efficient phase-dressed states} defined below, most of which are non-Gaussian.

\begin{definition}[Phase-dressed state]
\label{def:dressed}
Let $\ket{\psi}$ be an fGS of definite parity with computational-basis distribution $\pi$, and let $\theta:\{0,1\}^{d}\to\R$. The associated \emph{phase-dressed state} is
\begin{align}\label{eq:dressed}
  \ket{\psi_{\theta}}
  = U_{\theta} \ket{\psi}, \qquad U_{\theta}=\sum_{\x}e^{i\theta(\x)}\dyad{\x},
\end{align}
defined up to a global phase. We call $\ket{\psi_\theta}$ \emph{efficiently phase-dressed} if there is a classical algorithm that, on input $\x\in\{0,1\}^{d}$, computes $\theta(\x)$ in time $\mathrm{poly}(d)$.
\end{definition}

\noindent Since $U_\theta$ preserves each parity sector and leaves the distribution $\pi$ invariant, the certification walk of Eq.~\eqref{eq:certification_walk} remains unchanged. Within a fixed-parity sector the conditional state of the final two modes remains two-dimensional, so Lemma~\ref{lem:combined} still applies; the dressing only changes the relative phase of that conditional state, and hence only the basis in which the final mode is measured.

\begin{corollary}[Certification of phase-dressed states]
\label{cor:dressed}
Let $\ket{\psi_\theta}$ be an efficiently phase-dressed state as in Definition~\ref{def:dressed}. Algorithm~\ref{algo}, with the adaptive measurement basis computed from $\ket{\psi_\theta}$, certifies $\ket{\psi_\theta}$ to infidelity $\epsilon$ with confidence $1-\delta$ using
\begin{align}\label{eq:dressed-samples}
    N(\epsilon,\delta,d) \le \left\lceil \frac{d(d-1)}{\epsilon}\ln\frac{1}{\delta} \right\rceil
\end{align}
copies, single-qubit measurements only, and $\mathrm{poly}(d)$ classical processing per copy.
\end{corollary}

\begin{proof}
The verification operator depends on the target only through $\pi$, which is invariant under $U_\theta$. Thus, the spectral gap bound of Theorem~\ref{thm:main} applies verbatim. The classical processing computes the Gaussian conditional amplitudes at cost $O(d^{3})$ as before, and additionally evaluates $\theta$ on the observed outcome string at cost $\mathrm{poly}(d)$ by Definition~\ref{def:dressed}.
\end{proof}

\paragraph*{Phase dressing and non-Gaussianity.} Phase-dressing does not always inject non-Gaussianity. For $\theta(\x)=\sum_{j}\varphi_{j}x_{j}$ one has $U_{\theta}=\exp(i\sum_{j}\varphi_{j}a_{j}^{\dagger}a_{j})$, a passive FLO transformation (single-qubit $Z$ rotations in the qubit picture), so $\ket{\psi_\theta}$ remains Gaussian. 
For $\theta(\x)=\sum_{j<k}\varphi_{jk}x_{j}x_{k}$ one obtains a product of controlled-phase operations,
\begin{align}\label{eq:cphase}
    U_{\theta}=\exp\Big(i\sum_{j<k}\varphi_{jk}\,n_{j}n_{k}\Big)
    = \prod_{j<k}\mathrm{CP}_{jk}(\varphi_{jk}),
\end{align}
which is quartic in the fermionic operators and therefore outside the group of FLO transformation.The four-mode family
\begin{align}\label{eq:hebenstreit}
    \ket{\Psi_\theta} = \frac{ \ket{0000} + \ket{0011} + \ket{1100} + e^{i\theta}\ket{1111} }{2},
\end{align}
is an example. $\ket{\Psi_{\theta}}$ is non-Gaussian for all $\theta\neq0 \pmod{2\pi}$.
Every pure fermionic non-Gaussian state is a magic state for matchgate computation \cite{hebenstreit2019all}, and on four modes every such state of even parity is equivalent, under an FLO transformation, to a member of~\eqref{eq:hebenstreit} \cite{hebenstreit2019all}.

\begin{figure}[tb]
    \centering
    \includegraphics{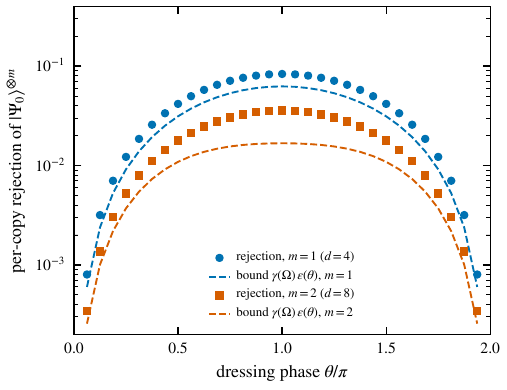}
    \caption{\label{fig:impostor}
    \textbf{Detection of a Gaussian impostor.} Per-copy probability that Algorithm~\ref{algo}, run with target $\ket{\Psi_\theta}^{\otimes m}$ of Eq.~\eqref{eq:hebenstreit}, rejects the Gaussian state
    $\ket{\Psi_{0}}^{\otimes m}$, for $m=1$ ($d=4$) and $m=2$ ($d=8$). Markers are exact values of $1-\bra{\Psi_0}^{\otimes m}\Omega\ket{\Psi_0}^{\otimes m}$.
    Dashed lines are the guaranteed lower bound $\gamma(\Omega)\,\epsilon(\theta)$ with $\gamma(\Omega)=1/d(d-1)$ and $\epsilon(\theta)$ from
    Eq.~\eqref{eq:impostor-infid}.  The rejection probability vanishes only at $\theta=0$ for which the state is Gaussian. Rejection is per copy and compounds to $1-(1-p)^{N}$ over $N$ copies.}
\end{figure}

The Gaussian member $\ket{\Psi_0}$ is a doublet of BCS states \eqref{eq:pair-example} with $u_r = v_r = 1/\sqrt{2}$, and the rest of the family is generated by a single controlled-phase gate acting across the two pairs,
\begin{align}
    \ket{\Psi_\theta} &= \mathrm{CP}_{jk}(\theta) \ket{\Psi_{0}},
    \qquad j\in\{1,2\},\ k\in\{3,4\}. 
\end{align}
$\ket{\Psi_{\theta}}$ is efficiently phase-dressed in the sense of Definition~\eqref{def:dressed} and its measurement distribution is uniform on $\{0000,0011,1100,1111\}$ for all $\theta$.
Corollary~\ref{cor:dressed} thus certifies the entire family \eqref{eq:hebenstreit}, and tensor products thereof, with the same sample complexity as in the Gaussian case. The protocol differs from the Gaussian case only in the final adaptive basis; conditioned on outcome $\ket{11}$ on the first two modes and $x\in\{+,-\}$ on the third, the final mode is measured in $(\ket{0}\pm e^{i\theta}\ket{1})/\sqrt{2}$ in place of $(\ket{0}\pm\ket{1})/\sqrt{2}$.

Figure~\ref{fig:impostor} shows the single-copy rejection probability of Algorithm~\ref{algo} when the target is $\ket{\Psi_\theta}^{\otimes m}$ and the lab state is the Gaussian impostor $\ket{\Psi_{0}}^{\otimes m}$. 
Since the two states are indistinguishable by any computational-basis measurement, the rejection is due entirely to the adaptive final measurement.
The rejection probability is bounded below by $\gamma(\Omega)\,\epsilon(\theta)$ (Eq.~\eqref{eq:plm}), 
where
\begin{align}\label{eq:impostor-infid}
    \epsilon(\theta) = 1 - \abs{\av{\Psi_{0}|\Psi_{\theta}}}^2
    = 1-\frac{10+6\cos\theta}{16}
    = \tfrac{3}{4}\sin^{2}(\theta/2),
\end{align}
and it exceeds this bound throughout, vanishing only at $\theta= 0$, where target and impostor coincide.

Not every non-Gaussian state falls within the phase-dressed class. The state $(\ket{0011} + e^{i\theta} \ket{1100})/\sqrt{2}$ utilized in a fermionic model of universal quantum computation and fermionic quantum advantage schemes \cite{bravyi_universal_2006,Ivanov2017,oszmaniec2022sampling}, for instance, cannot be written as a phase dressing of any Gaussian state.
A certification walk tailored to this state would require transitions between strings that differ in all four bits, hence a four-site Glauber dynamics rather than the two-site Glauber dynamics underlying our down-up walk. We leave open the question whether the zero-free argument extends to this setting.

%-------------------------------------
\section{Summary and Discussion}\label{sec:conclusion}
%-------------------------------------

We have given a protocol that certifies any $d$-mode pure fermionic Gaussian state that is both sample-efficient and computationally efficient, using only single-qubit measurements of which only one is adaptive. In particular, it requires
\begin{align}
    N(\epsilon,\delta,d) \le \left\lceil \frac{d(d-1)}{\epsilon}\ln\frac{1}{\delta} \right\rceil 
\end{align}
copies to certify to infidelity $\epsilon$ with confidence $1-\delta$. The scaling is optimal in $\epsilon$ and $\delta$. The $O(d^2)$ scaling and the single adaptive qubit places our protocol squarely in the middle between the decision-tree-basis protocol \cite{gupta2026few}, which consumes $O(d\epsilon^{-1})$ copies at the cost of measuring all $d$ qubits adaptively, and the estimation-based shadow-overlap protocol \cite{huang2025certifying}, which removes adaptivity entirely at the cost of returning to $O(d^4\epsilon^{-2})$.
Whether the three protocols occupy optimal points on the trade-off between $d$-dependence and adaptivity, or whether intermediate protocols can improve on both simultaneously, is posed as an open problem below.

The protocol can be thought of as a variant of the shadow-overlap protocol \cite{huang2025certifying} adapted to fermionic parity superselection. We mostly measure in the computational basis, with a parity test certifying that the lab state lies in the correct parity sector, and a single $X$ measurement on the penultimate qubit before the final adaptive step. The $X$ measurement is what makes the protocol work: within a fixed-parity sector, the conditional state of the final two modes is two-dimensional, and measuring one of them in the $X$ basis realizes exactly the rank-one projector onto the conditional two-mode target state using only single-qubit measurements. The resulting verification operator maps to the certification walk whose rapid mixing implies efficient certification, in parallel with the analysis of Ref.~\cite{huang2025certifying}.

The main departure from Ref.~\cite{huang2025certifying} is in how soundness is established. Their proof of rapid mixing invokes a Haar-typicality argument that does not apply to the structured class of Gaussian states. To obtain a guarantee for \emph{all} pure Gaussian states, we instead establish that the generating polynomial of the stationary distribution of the walk---the computational-basis distribution $\pi(\x) = \abs{\av{\x|\psi}}^2$ of the target---is zero-free in the open right half-plane. This zero-free region, which follows from the Pfaffian structure of $\pi$ and the orthogonality $\Gamma^{\T}\Gamma = \id$ that characterizes pure Gaussian states, implies the spectral-gap bound via results connecting zero-free regions to the relaxation of multi-level down-up walks \cite{anari2019logII,anari2021entropic,alimohammadi2021fractional_arxiv}, tools that emerged only recently from the theory of high-dimensional expanders and the study of counting and sampling problems in combinatorics. Their appearance here opens, in our view, a new point of contact between two previously separate programs: the framework of analyzing state certification via Markov chains \cite{huang2025certifying}, and the spectral and entropic independence techniques of \cite{anari2019logII,anari2021entropic,alimohammadi2021fractional_arxiv}. Whether this connection is an isolated accident or a more systematic bridge is, we think, the most interesting question this work raises.

We now list several open directions.

\paragraph*{Trade-off between $d$-dependence and adaptivity.}\label{sec:open:dtradeoff} As summarized above, our protocol sits between the shadow-overlap protocol \cite{huang2025certifying} and the decision-tree-basis protocol \cite{gupta2026few} on the trade-off between $d$-dependence and adaptivity, using a single adaptive qubit to attain $\epsilon^{-1}$ scaling at $O(d^2)$ copies. Whether the $d^2$ prefactor can be reduced without increasing the number of adaptive measurements, or conversely whether the single adaptive qubit can be eliminated at a smaller cost than the $d^2\epsilon^{-1}$  factor paid by the shadow-overlap protocol, whose non-adaptive variant costs $O(d^4\epsilon^{-2})$, are both open. 
What one would hope for is a unified lower bound that captures the trade-off between $d$-dependence and the number of adaptive measurements, clarifying whether the three protocols above occupy optimal points on the trade-off curve or whether intermediate protocols can do better on both axes simultaneously.

\paragraph*{Exploration of proof techniques} The proof of Theorem~\ref{thm:main} naturally suggests two complementary directions, seemingly in opposite spirits.
On the one hand, the zero-freeness argument and the broader toolkit of the spectral and entropic independence program \cite{anari2019logII,anari2021entropic,alimohammadi2021fractional_arxiv}, which offers several routes to establishing rapid relaxation beyond zero-freeness alone, may apply to other structured families of quantum states. Whether this is the case is wide open, and identifying which state families admit a zero-free generating polynomial seems to us a natural entry point. 
On the other hand, the current proof imports tools from combinatorics and theoretical computer science that, while powerful, are not part of the standard toolkit of quantum many-body physics. Developing a more physics-native proof of the same result that works directly with covariance matrices, Wick's theorem, and fermionic correlation functions
might give more intuition about which properties of more general quantum states allow efficient certification,
pointing toward certification protocols for other state families.

\paragraph*{Average-case spectral gap for Haar-random fGS.} Numerics up to $d = 14$ (Fig.~\ref{fig:gap}) show that Haar-random Gaussian states exhibit a spectral gap scaling as $\Theta(d^{-1})$, a factor of $d$ better than the worst-case bound. An analytic proof of this average-case scaling would imply that the sample complexity for certifying a typical fGS is $O(d\epsilon^{-1}\ln\delta^{-1})$, matching the worst-case sample complexity of the decision-tree-basis protocol of \cite{gupta2026few}, while using only one adaptive measurement rather than $d-1$. 
As the natural route to prove rapid relaxation on average of Ref.\cite{huang2025certifying} does not transfer here, as discussed in Sec.~\ref{sec:related}, a new proof technique seems necessary to establish the average-case spectral gap.

\paragraph*{Certifying mixed and thermal Gaussian states.} The class of fermionic Gaussian states is not confined to pure states \cite{bravyi2004lagrangian}, but our analysis relies on purity in an essential way. The spectral-gap bound rests on $\S_1$-sector stability of the generating polynomial $Z_\pi$, which in turn requires the covariance matrix to satisfy $\Gamma^{\T}\Gamma = \id$, a condition that holds if and only if the Gaussian state is pure. 
For a mixed Gaussian state the singular values of $\Gamma$ are strictly less than one, and the best sector stability one could hope for is $\mathcal{S}_\alpha$ for some $\alpha < 1$. Such a result would not connect to the $d\leftrightarrow d-2$ down-up walk; $\mathcal{S}_\alpha$-stability with $\alpha < 1$ implies rapid mixing only for $d\leftrightarrow d-l$ walks with step size $l > 2$, which would require measuring more than one qubit adaptively. Extending the protocol to mixed targets would therefore require either a new proof strategy or a substantially different measurement scheme.

\paragraph*{Certification in the presence of correlated preparations and adversaries.}
Our protocol assumes that the lab-prepared states are independent across rounds.
In practice, correlations can arise from drift, crosstalk, or a malicious device, and a stronger guarantee against an adversary who produces an arbitrary correlated or entangled state across all $N$ copies is desirable.
The adversarial framework of Zhu and Hayashi \cite{zhu2019framework} provides a general recipe called \emph{hedging} to convert a non-adversarial verification operator $\Omega$ to an adversarially-robust one, $\Omega_{p^*} = (1-p^*)\Omega + p^*\id$, with some suitable probability $p^*$; the hedged protocol $\Omega_{p^*}$ is realized simply by flipping a coin before each round and accepting unconditionally with probability $p^*$.
Hedging comes with an overhead in sample complexity that depends not only on the spectral gap, but also the smallest eigenvalue $\lambda_{\textrm{min}}(\Omega)$. We leave a detailed investigation into this construction as a future direction.

\paragraph*{Robustness.} This work considers the common QSV framework in which the null hypothesis corresponds to the lab state being \emph{exactly} the target: $\expval{\rho}{\psi}=1$. In practice, no preparation is perfect, and a more operationally meaningful test would accept any state within an $\epsilon'$-ball of the target.
Formally, a \emph{robust} certification protocol distinguishes  $\expval{\rho}{\psi}\ge 1-\epsilon'$ from $\expval{\rho}{\psi}\le 1-\epsilon$ for some constants $0 < \epsilon' < \epsilon$ with confidence $1-\delta$.
To our knowledge, the only certification protocol achieving this particular notion of robustness with single-qubit measurements is the recent work of Coladangelo, Li, and Slote \cite{coladangelo2026robust}. Like the shadow-overlap protocol \cite{huang2025certifying}, their guarantee is for Haar-typical states, but their oracle is stronger, requiring amplitude access in both the computational basis and the Hadamard (Pauli-$X$) basis.
A robust certification protocol with a worst-case guarantee for any pure fGS would be of practical interest.

For phase-dressed states,
approximating an output probability in the Hadamard basis to multiplicative error is at least as hard as approximating the same task for an IQP output probability, which is \#P-hard \cite{bremner2016average}.
To see this, let $\ket{\psi}$ be an equal superposition over a fixed-parity sector, which is Gaussian (see, e.g., \cite[Lemma 16]{dias2026gate-decomposition}). 
Applying a diagonal phase-dressing $U_{\theta}$ to this state produces the output state of a general IQP circuit except without the final Hadamard layer. 
Measuring $U_{\theta}\ket{\psi}$ in the Hadamard basis is therefore equivalent to measuring the corresponding IQP circuit in the computational basis, so estimating the outcome probability is \#P-hard. Whether this hardness result genuinely obstructs oracle-based certification in the style of \cite{coladangelo2026robust} for phase-dressed states is left open.

Another route to robustness is to convert the protocol into a fidelity estimator from the pass/fail frequency \cite{zhu2019framework}. (We note that robustness in this sense is incomparable to that of \cite{coladangelo2026robust} discussed above as the two notions address fundamentally different tasks. In particular, estimation provides no guarantee that states within a prescribed ball are accepted with certainty.) 
However, this route requires controlling the smallest eigenvalue $\lambda_{\textrm{min}}(\Omega)$ in addition to the spectral gap and the estimation is tight only when $\lambda_{\textrm{min}} (\Omega) = \lambda_2 (\Omega)$, i.e., when the protocol is \emph{homogeneous}. Ours is not homogeneous, but even setting that aside, this estimation route cannot be competitive on sample complexity. For homogeneous protocols, the fidelity estimator requires \cite{zhu2019framework}
\begin{align}
    N \;\le\; \frac{1}{4\,\gamma(\Omega)^2 \epsilon^2}
\end{align}
samples to estimating the fidelity to additive error $\epsilon$, which gives $N=O(d^4\epsilon^{-2})$ for our protocol, matching that of the estimation-based variant of the shadow-overlap protocol \cite{huang2025certifying}, and a factor of $d$ worse than $O(d^3\epsilon^{-2})$ of the fidelity witness \cite{gluza2018fidelity} from the improved analysis in Sec.~\ref{sec:related}.

\paragraph*{Extension to other non-Gaussian states.} The phase-dressed states of Sec.~\ref{sec:dressed} already extend our protocol beyond Gaussian states. However, a more substantial extension would be to states produced by a bounded number of non-Gaussian gates of the kind studied in fermionic learning algorithms \cite{mele2025few_fermionic,reardon2024improved}.
Another interesting direction concerns certification of outputs of FLO circuits.
Given classical descriptions of an FLO unitary $U$ and a non-Gaussian magic input state $\ket{\phi}$, the task is to certify that the device produced $U\ket{\phi}$. This is directly relevant to validating the output of the fermionic quantum advantage scheme \cite{oszmaniec2022sampling}.

\paragraph*{Genuinely fermionic certification protocol.} Recently, Ref.~\cite{koizumi2026provably} posed the question of whether certification of fermionic states can inherit the advantage of number-conserving shadow tomography. For $d$-mode state with fixed particle number $n$, all $k$-RDMs can be estimated using  $O(n^k \epsilon^{-2})$ samples independent of $d$. (The case $k=1$ is relevant for Slater determinants.) 
Our qubit-based protocol does not directly address this question, since it breaks parity superselection by measuring in the $X$-basis. The complexity of a genuinely fermionic certification protocol, i.e., one that employs only parity-conserving operations, is therefore still an open problem.

%-------------------------------------
\section*{Acknowledgment}
%-------------------------------------

ND thanks Chulalongkorn university and Siam Quantum Square for their hospitality during his visits, and thanks Marek Gluza for encouraging discussion. The authors performed AI-assisted numerical experiments to determine the scaling of the spectral gap of $P$ in Eq.~\eqref{eq:certification_walk}, including the absolute constant, which was then fed to Claude Fable 5 to produce a one-shot proof of Theorem~\ref{thm:main} in the number-conserving case. Later, Claude was able to extend the proof to the general Gaussian case.
The authors verified both proofs, revised, and substantially expanded the exposition into Appendices~\ref{app:markov} to \ref{app:sub-proof}. Interestingly, Claude gave the key identity~\eqref{eq:CayleyIdentity} without a derivation; the proof of Proposition~\ref{prop:CayleyTransform} using Grassmann integrals is our own. In addition, we used GPT-6 Astra to improve the fidelity witness bound in Sec.~\ref{sec:related}. All other parts of the work are done by us (human).

This research has received funding support from the NSRF via the Program Management Unit for Human Resources \& Institutional Development, Research and Innovation [grant number B39G690073]. 

\bibliography{ref}

\clearpage
\appendix

%==========================
\section{Grassmann--Berezin integral representation of the Pfaffian}\label{app:grassmann}
%==========================
Pfaffians give compact expressions for state overlaps and computational-basis probabilities of fermionic Gaussian states in terms of their covariance matrices $\Gamma$. Here we review the relevant identities, then introduce the Grassmann--Berezin integral representation of the Pfaffian and its application to overlaps of fermionic Gaussian states. The same representation is used in the spectral-gap proof to analyze the generating polynomial of the computational-basis distribution.

Let $N=2d$ be even and $\{e_j\}_{j=1}^N$ be a canonical basis for $\C^N$.
Given a complex antisymmetric matrix
$A\in\mathbb{C}^{N\times N}$ and its associated 2-form 
\begin{align}
    \omega_A = \frac{1}{2} \sum_{j,k} A_{jk} e_j \wedge e_k,
\end{align}
the Pfaffian is defined as the top coefficient
\begin{align}
    \frac{1}{d!} \omega_A^{\wedge d} = \Pf(A) \,\vol,
\end{align}
once one has chosen a volume form.
The Pfaffian is defined to be zero for all odd-dimensional matrices.

Two well-known identities are
\begin{align}
    \Pf(A)^2 &= \det(A), \label{eq:pfaffian-det}\\
    \Pf(RAR^{\T}) &= \det(R)\,\Pf(A), \label{eq:pfaffian-basis-change}
\end{align}
for any $R\in\mathbb{C}^{N\times N}$.
Eq.~\eqref{eq:pfaffian-det} allows the Pfaffian to be computed up to a sign in $O(d^3)$ time. This is sufficient for computing the computational-basis probabilities in Eq.~\eqref{eq:PfaffianProb},
\begin{align}
    \pi(\vb{x})
    = \abs{\av{\vb{x}|\psi}}^2
    = \frac{1}{2^d}
    \abs{\Pf\bigl(\Gamma+D_{\x}J\bigr)},
\end{align}
which depend only on the absolute value of the Pfaffian.
In contrast to the determinant, the Pfaffian is basis-\emph{dependent} as evident from Eq.~\eqref{eq:pfaffian-basis-change}. Under an orthogonal change of basis $R$, the Pfaffian acquires a factor of $\det(R)=\pm1$ depending on whether $R$ preserves or reverses the orientation. In particular, if $R$ is a permutation of basis vectors, $\det(R)$ is the sign of the permutation. 

Changing the orientation of the basis corresponds to changing the volume form by a sign, which changes the Pfaffian by the same sign.
In the fermionic setting, this choice of sign is equivalent to a choice of ordering for the Majorana operators. In the ``split ordering" that we use, in which all odd (resp. even) Majorana operators are grouped together, the covariance matrix of the Fock vacuum $\ket{0_F}$ is
\begin{align}
    J = \mqty(0&\id_d\\-\id_d&0),
\end{align}
whereas in the more common ``interleaved ordering", in which $c_j$ and $c_{d+j}$ are placed next to each other, the covariance matrix of the vacuum is 
\begin{align}
    J' = \mqty(0&1\\-1&0)^{\!\oplus d}.
\end{align}
The choice amounts to requiring either $J$ or $J'$ to have unit Pfaffian \cite{caracciolo2013pfaffians}.
The two matrices are related by a permutation of the basis vectors with $\sum_{j=1}^d (j-1) = d(d-1)/2$ transpositions. Hence, 
\begin{align}
    \Pf(J)
    =(-1)^{d(d-1)/2} \, \Pf(J').
\end{align}

In the interleaved ordering for which $\Pf(J')=1$, the Pfaffian is given explicitly by the polynomial
\begin{align}\label{def:pfaffian}
    \Pf(A)
    = \frac{1}{2^d d!}
    \sum_{\tau\in S_N}
    \textrm{sgn}(\tau)\,
    A_{\tau(1),\tau(2)}
    \cdots
    A_{\tau(N-1),\tau(N)},
\end{align}
where $S_N$ denotes the permutation group on $N$ elements and $\textrm{sgn}(\tau)$ is the sign of a permutation $\tau$.
For example, for $N=4$, Eq.~\eqref{def:pfaffian} reduces to
\begin{align}
    \Pf(A)
    = A_{12}A_{34}-A_{13}A_{24}+A_{14}A_{23},
\end{align}
from which one readily verifies that
\begin{align}
    \Pf(J)=-1, \qquad
    \Pf(J')=1.
\end{align}
With this convention, the Pfaffian of the covariance matrix is precisely the parity of an fGS \cite{bravyi2017impurity}: denoting the eigenvalues of the parity operator $\texttt{Par} = \prod_{j=1}^d (-ic_j c_{d+j})$ by $\sigma \in \{\pm1\}$, we have
\begin{align}
    \Pf (\Gamma) = \sigma.
\end{align}

The Pfaffian admits a useful representation as a Gaussian integral of  anticommuting Grassmann variables \cite{bravyi2004lagrangian,bravyi2017impurity,caracciolo2013pfaffians}. For an introduction to integration over anticommuting variables, see Ref.~\cite{efetov1997supersymmetry}; a pedagogical discussion of the Berezin rules and Gaussian determinant representations can be seen in Ref.~\cite[Sec.~3.2]{chotibut2026random}. 

A complex Grassmann algebra $\mathcal G_N$ is generated by $\{\theta_j\}_{j=1}^N$ satisfying
\begin{align}
    \theta_j \theta_k = -\theta_k \theta_j.
\end{align}
Due to the anticommutativity, $\theta_j^2=0$ for all $j$, which makes $\mathcal G_N$ extremely simple: an arbitrary element $f\in\mathcal G_N$ can be expressed as a complex linear combination of Majorana monomials (or Majorana configurations),
\begin{align}
    f(\theta) = \sum_{X \subseteq [N]} \alpha_X \,\theta^X, \qquad \theta^X \coloneqq \theta_1^{x_1} \cdots \theta_N^{x_N}. 
\end{align}
Analytic functions of Grassmann variables are defined through their Taylor series, which truncates after the linear term. For example, the exponential function reduces to 
\begin{align}\label{eq:GrassmannExp}
\exp(\theta^X)=1+\theta^X.
\end{align}

A Grassmann-Berezin integral is defined as a linear functional that extracts the coefficient of the variable being integrated over. In particular, the integral is completely specified by its action on a linear Grassmann polynomial,
\begin{align}\label{eq:grassmann-integral}
    \int d\theta \, (a + b\theta) = b.
\end{align}
Thus, Grassmann integration is algebraically the same as differentiation,
\begin{align}
    \int d\theta_j 
    \equiv
    \frac{\partial}{\partial\theta_j}: \mathcal G_N \to \mathcal G_{N-1}, 
\end{align}
and should not be thought of as a Riemann integration. 
The vanishing of the constant term can be understood from the requirement of translational invariance of the integration measure,
\begin{align}
    \int d\theta f(\theta + \eta) = \int d\theta f(\theta),
\end{align}
where $\eta$ is an independent Grassmann variable. Taking $f(\theta)=\theta$ gives
\begin{align}
    \int d\theta \, \theta &= \int d\theta (\theta + \eta) \\
    &= \int d\theta \, \theta + \eta \int d\theta \, 1
\end{align}
Therefore, $\int d\theta 1 =0$. Fixing the normalization, $\int d\theta \, \theta=1$, then yields the rule \eqref{eq:grassmann-integral}.
With the convention
\begin{align}
    \int D\theta \: \theta^X =
    \begin{cases}
        1, & X=[2d],\\
        0, & \text{otherwise},
    \end{cases}
\end{align}
equivalent to choosing the integration measure 
\begin{align}
D\theta \coloneqq d\theta_{N}\cdots d\theta_1, 
\end{align}
the fundamental Gaussian integrals are \cite{caracciolo2013pfaffians}
\begin{align}
    \int D\theta\, \exp\left( \frac{1}{2}\theta^{\T} M \theta \right) &= \Pf(M), \\
    \int D\eta\, \exp\left( \theta^{\T}\eta +  \frac{1}{2}\eta^{\T} M \eta \right) &= \Pf(M) \exp\left(\frac{1}{2} \theta^{\T} M\,\inv \theta \right),
\end{align}
where we consider $\theta=(\theta_1, \dots, \theta_N)$ as a column vector and $\theta^{\T}$ its transpose.

The Grassmann representation provides a convenient framework to compute quantities associated to Gaussian states.
Any linear operator $B$ on the Fock space can be expanded in the basis of Majorana monomials,
\begin{align}
    B = \sum_{X \subseteq [N]} B_X c^X, \qquad
    c^X \coloneqq c_1^{x_1} \cdots c_{N}^{x_{N}}.
\end{align}
The expansion naturally associate with $B$ the Grassmann generating polynomial
\begin{align}
    B(\theta) = \sum_X B_X \theta^X.
\end{align}
The trace of a product of operators can then be expressed as a Grassmann integral
\begin{align}\label{eq:grassmann-trace-inner-product}
    \Tr(AB) = \frac{1}{2^d} \int D(\theta) D(\eta) A(\theta) B(\eta) \exp({\theta^{\T}\eta}).
\end{align}
The Grassmann generating polynomial of an fGS $\Psi=\dyad{\psi}$ with covariance matrix $\Gamma$ takes the Gaussian form
\begin{align}
    \Psi(\theta) =
    \frac{1}{2^d}
    \exp\left(
        -\frac{\mathrm{i}}{2}\theta^{\T}\Gamma\theta
    \right).
\end{align}
By using Eq.~\eqref{eq:grassmann-trace-inner-product}, the overlap between two pure Gaussian states $\Psi$ and $\Phi=\dyad{\phi}$ can be expressed as \cite{bravyi2017impurity}
\begin{align}
    \abs{\av{\psi|\phi}}^2 = \Tr(\Psi \Phi) 
    &= \frac{1}{2^d} \,\Pf(\Gamma_{\phi})\, \Pf(\Gamma_{\psi} + \Gamma_{\phi}) \nonumber\\ 
    &= \frac{\sigma}{2^d}\, \Pf(\Gamma_{\psi} + \Gamma_{\phi}), 
\end{align}
where in the last line we have specialized to the interleaved ordering of Majorana operators for which $\Pf(\Gamma)=\sigma$. Note that $\sigma$ here is the parity of either state, as states of opposite parity are orthogonal.

%==========================
\section{Verification operator as the certification walk}\label{app:verify-op-is-markov}
%==========================

In this appendix, we show that the verification operator
\begin{align}
  L = \binom{d}{2}^{\!\!-1} \sum_{k}\sum_{\z^{(k)}} \dyad{\smash{\z^{(k)}}} \otimes \dyad{\psi_{\z^{(k)}}}
\end{align}
has the same spectrum as the transition matrix of a Markov chain we refer to as the certification walk.

Suppose we write the target Gaussian state in the computational basis as
\begin{align}
    \ket{\psi} = \sum_{\x} \sqrt{\pi(\x)} \: e^{i \phi(\x)} \ket{\x}.
\end{align}
For each branch $\z^{(k)}$ of $E_k$,
if $\x$ and $\y$ are compatible with the same branch, then they agree on all qubits in $[d]\backslash k$. Since the target has definite parity, the two states must either coincide or differ on both qubits in $k$, cf. Eq.~\eqref{eq:conditional-two-qubit}. That is, either $\x=\y$ or $\abs{\x \oplus \y}=2$.

It is convenient to consider the support 
\begin{align}
    V = \{\x \mid \pi(\x) >0\}
\end{align}
as the vertex set of a graph in which two vertices $\x,\y \in V$ are connected, $(\x,\y)\in E$, if and only if they have Hamming distance two, i.e., $|\x \oplus \y|=2$. 
When the distribution has full support on the relevant parity or particle-number sector, this graph is the corresponding full induced subgraph of the hypercube. In particular, for a general pure fGS with fixed parity, $G=(V,E)$ is the halved-cube graph $\frac{1}{2}Q_d$, while for an fGS with fixed particle number $n$, it is the Johnson graph $J(d,n)$. If the distribution does not have full support on the relevant sector, $G$ is instead an induced subgraph of the corresponding halved-cube or Johnson graph.
Alternatively, identifying each bitstring $\x$ with a subset $X \subseteq [d]$ via $j\in X$ if and only if $x_j=1$, the condition $|\x \oplus \y|=2$ corresponds to moving between subsets whose symmetric difference has size two: 
\begin{align}
    \abs{X\oplus Y} \coloneqq  \abs{(X\backslash Y) \cup (Y\backslash X) } =2.
\end{align}
Both viewpoints will be useful when interpreting the transitions as a so-called multi-level down-up walks.

We begin with the following lemma.

\begin{lemma}
The matrix elements of $L$ are given by
    \begin{align}
\mel{\x}{L}{\y} = 
    \begin{cases}
        {\displaystyle {d \choose 2}^{\!\!-1} \frac{ \sqrt{\pi(\x) \pi(\y)} }{\pi(\x)+\pi(\y)} \: e^{i(\phi(\x) - \phi(\y))} }, 
        &(\x,\y)\in E,  \\
        {\displaystyle
        {d \choose 2}^{\!\!-1} \mkern-9mu \sum_{\substack{\x' \\ |\x \oplus \x'|=2}} \mkern-9mu \frac{\pi(\x)}{\pi(\x) + \pi(\x')} }, 
        &\x=\y, \\
        0, &\textrm{otherwise.}
    \end{cases}
\end{align}
\end{lemma}
\noindent The matrix elements have the same form as those in Eq.~(13) of Ref.~\cite{huang2025certifying}, with the only difference being the graph structure. In the level-2 protocol of Ref.~\cite{huang2025certifying}, there would also be transitions between vertices of Hamming distance one. Here, such transitions are forbidden by the fermionic parity superselection rule.

\begin{proof}
\textbf{Case I: $(\x,\y)\in E$.}
Writing
\begin{align}
\x=(\z^{(k)},b_1b_2),
\qquad
\y=(\z^{(k)},\overline b_1\overline b_2),
\end{align}
the conditional two-qubit state, Eq.~\eqref{eq:conditional-two-qubit}, takes the form
\begin{align}
\ket{\psi_{\z^{(k)}}}
= \frac{\sqrt{\pi(\x)}\: e^{i \phi(\x)}  \ket{b_1b_2}
+ \sqrt{\pi(\y)}\: e^{i \phi(\y)} \ket{\smash{\overline b_1\overline b_2}} 
}{\sqrt{\pi(\x) + \pi(\y)}}.
\end{align}
Consequently, in the two-dimensional subspace spanned by $\ket{\x}$ and $\ket{\y}$, the corresponding branch of $E_k$ is proportional to
\begin{align}
\mqty(
\pi(\x) & \sqrt{\pi(\x)\pi(\y)} e^{i(\phi(\x)-\phi(\y))} \\
\sqrt{\pi(\x)\pi(\y)} e^{-i(\phi(\x)-\phi(\y))}  & \pi(\y))
\end{align}
with the normalization factor $(\pi(\x)+\pi(\y))^{-1}$.
It follows that
\begin{align}
    \mel{\x}{L}{\y} = {d \choose 2}^{\!\!-1} \frac{ \sqrt{\pi(\x) \pi(\y)} }{\pi(\x)+\pi(\y)} \: e^{i(\phi(\x) - \phi(\y))} 
\end{align}
whenever $(\x,\y)\in E$.

\noindent\textbf{Case II: $\x=\y$.} The diagonal element  
$\mel{\x}{L}{\x}$ receives a contribution proportional to $\pi(\x)$ from every branch $\x'$ such that $|\x \oplus \x'|=2$ (in particular, $\x'$ needs not be in the support $V$) with a distinct normalization factor $(\pi(\x) + \pi(\x'))^{-1}$ for each branch.

\noindent\textbf{Case III.} In the remaining case when $\x$ or $\y$ does not agree with $\z_k$ on $[d]/k$ in the first place, the matrix element is zero.

\end{proof}

Following Ref.~\cite{huang2025certifying}, we can now remove the phases and rescale the basis to obtain a Markov transition matrix. 
By defining the diagonal transformations
\begin{align}
    S = \sum_{\x \in V} \pi(\x) \dyad{\x}, \qquad
    F = \sum_{\x \in V}  e^{i\phi(\x)} \dyad{\x},
\end{align}
the target state can be mapped to the all-ones vector via
\begin{align}
F\dgg S^{\numinv{\frac{1}{2}}}\!\ket{\psi} = \mqty(1&\cdots &1)^{\T}.
\end{align}
We therefore define
\begin{align}
    P =  F\dgg S^{\,\text{-}\hspace{-1pt}\frac{1}{2}} L S^{\,\hspace{-1pt}\frac{1}{2}} F. 
\end{align}
Here, $S\numinv{\frac{1}{2}}$ is understood as the inverse of $S^{\frac{1}{2}}$ on the support $V$, and as zero on its orthogonal complement.
Its matrix elements are
\begin{align}
\mel{\x}{P}{\y} = 
    \begin{cases}
        {\displaystyle {d \choose 2}^{\!\!-1} \frac{ \pi(\y) }{\pi(\x)+\pi(\y)}}, 
        & (\x,\y)\in E,  \\
        {\displaystyle
        {d \choose 2}^{\!\!-1} \mkern-9mu \sum_{\substack{\x' \\ |\x \oplus \x'|=2}} \mkern-9mu \frac{\pi(\x)}{\pi(\x) + \pi(\x')} }, 
        & \x=\y, \\
        0, & \textrm{otherwise.}
    \end{cases}
\end{align}
Notably, the phases $\phi(\x)$ cancel under the similarity transformation, so that $P$ depends only on the computational-basis probabilities $\pi$.
Since $\ket{\psi}$ is a 1-eigenvector of $L$, the all-ones vector is a right eigenvector of $P$ with eigenvalue 1, showing that $P$ is row-stochastic and can be interpreted as a transition matrix of a classical Markov chain we call the \emph{certification walk}.

%========================
\section{Essentials of Markov chains}\label{app:markov}
%------------------------

A Markov chain, or simply a (random) walk, on a state space $\mathcal X$ is a discrete-time stochastic process in which the next state depends only on the current state. It is fully specified by a row-stochastic transition matrix $P \in \R^{\mathcal X \times \mathcal X}$.
We will often use the terms ``walk" and ``transition matrix" interchangeably, referring to $P$ itself as the walk.
The evolution of a distribution $\nu_t: \mathcal X\to \R_{\ge0}$ over states is given by right multiplication, $\nu_{t+1}=\nu_tP$, or in coordinates,
\begin{align}
    \nu_{t+1}(x) = \sum_{y} \nu_t(y) P(y,x).
\end{align}
A distribution $\mu$ is called \emph{stationary} if $\mu P=\mu$. To verify that $\mu$ is stationary, it suffices to check the detailed balance condition:
\begin{align}\label{eq:detailed_balance}
    \mu(x) P(x,y) = \mu(y) P(y,x) \quad \text{for all } x,y \in \mathcal X.
\end{align}
If Eq.~\eqref{eq:detailed_balance} holds, we say that the walk is \emph{reversible} with respect to $\mu$, or simply that $P$ or $\mu$ is reversible when the meaning is clear from the context.

The transition matrix also acts on functions $f\in \R^{\mathcal X}$ via left multiplication:
\begin{align}
    (Pf)(x) = \sum_y P(x,y) f(y).
\end{align}
Applying $P$ to $f$ corresponds to taking the average of $f$ with respect to the distribution of the next step:
\begin{align}
\nu_t (P f) = \sum_x \nu_{t+1}(x) f(x) = \E_{\nu_{t+1}}[f].
\end{align}
In particular, if $\mu$ is stationary, then $\mu(Pf)=\E_{\mu}[f]$. 
Since $P$ is row-stochastic, the constant function $\mqty(1&\cdots&1)^{\T}$ is a right eigenfunction with eigenvalue $1$. Every eigenfunction with eigenvalue different from $1$ has zero mean with respect to the stationary distribution: $\E_{\mu}[f_j]=0$.

Reversible Markov chains have real eigenvalues lying in $[-1,1]$ by the Perron-Frobenius theorem (see, e.g., \cite[Lemma 12.1]{levin2017markov}). It is also common to make a walk \emph{lazy} (by adding self-loops with probability $1/2$) to eliminate negative eigenvalues.
In our case, the certification walk is reversible with respect to the distribution $\pi$, and its spectrum already lies in $[0,1]$ without laziness, as the transition matrix is related to the verification operator---a POVM element, hence positive semidefinite---by an similarity transformation that preserves the spectrum (Appendix~\ref{app:verify-op-is-markov}).

It is convenient when analyzing eigenfunctions, to introduce an inner product under which $P$ becomes self-adjoint. Suppose that $P$ is reversible with respect to $\mu$, and define
\begin{align}
\av{f, g}_{\mu} \coloneqq \sum_x \mu(x) f(x) g(x).
\end{align}
Then one verifies that $P$ is self-adjoint under this inner product:
$\av{f,Pg}_{\mu} = \av{Pf,g}_{\mu}$.
Equivalently, by defining the diagonal matrix $S = \mathrm{diag}(\mu)$, reversibility is the statement $SP = P^\T S$, which implies that
\begin{align}
    A \coloneqq S^{\frac{1}{2}} P S\numinv{\frac{1}{2}} 
\end{align}
is symmetric.
In particular, $A$ admits a complete orthonormal basis of eigenfunctions $\{\varphi_j\}_{j=1}^{\abs{\mathcal X}}$, and each $f_j \coloneqq S\numinv{\frac{1}{2}} \varphi_j$ is an eigenfunction of $P$ with the same eigenvalue.

The eigenfunctions $f_j$ form an orthonormal basis with respect to $\av{\cdot,\cdot}_{\mu}$, and we choose $f_1=1$.
For an initial distribution $\nu_0$ and its evolution $\nu_t=\nu_0P^t$, reversibility gives the spectral expansion
\begin{align}
\frac{\nu_t(x)}{\mu(x)}-1
&=\sum_{j=2}^{|\mathcal X|}a_j\lambda_j^t f_j(x),
\label{eq:spectral_expansion}\\
a_j&=\sum_x\nu_0(x)f_j(x).\nonumber
\end{align}
Assume that the spectrum is nonnegative and ordered as $1=\lambda_1>\lambda_2\ge\cdots\ge\lambda_{|\mathcal X|}\ge0$. Writing $\gamma=1-\lambda_2$ for the spectral gap and $\|g\|_\mu^2=\av{g,g}_\mu$, we obtain
\begin{align}
\left\|\frac{\nu_t}{\mu}-1\right\|_\mu
&\le (1-\gamma)^t
\left\|\frac{\nu_0}{\mu}-1\right\|_\mu\nonumber\\
&\le e^{-\gamma t}
\left\|\frac{\nu_0}{\mu}-1\right\|_\mu.
\label{eq:spectral_relaxation}
\end{align}
Thus the spectral gap controls relaxation to the stationary distribution $\mu$: the slowest nonstationary mode decays as $\lambda_2^t$, and the relaxation time is defined by
\begin{align}
\tau_{\mathrm{relax}}\coloneqq\frac{1}{\gamma}.
\label{eq:relaxation_time}
\end{align}

The \emph{Dirichlet form} of a walk $P$ reversible with respect to $\mu$ is defined as
\begin{align}\label{def:dirichlet}
\dirich{f,g} &\coloneqq \av{(\id - P)f,g}_{\mu} = \av{f,(\id-P)g}_{\mu},
\end{align}
with a useful identity \cite[Lemma 13.6]{levin2017markov},
\begin{align}\label{eq:dirichletDiag}
    \dirich{f,f} = \frac{1}{2} \sum_{x,y} (f(x)-f(y))^2 \mu(x) P(x,y).
\end{align}
The smallest eigenvalue of $\id-P$ on the subspace of mean-zero functions is the spectral gap $\gamma$. Consequently, for every nonzero $f$ with $\E_\mu[f]=0$, we have the bound
\begin{align}
    \gamma \le \frac{\dirich{f,f}}{\av{f,f}_{\mu}}
    = \frac{\dirich{f,f}}{\E_{\mu}[f^2]},
\end{align}
which leads to the variational characterization
\begin{align}
    \gamma = \inf_{\substack{f\\ \E_{\mu}[f]=0 \\ \Var_{\mu}[f] \neq 0}} \frac{\dirich{f,f}}{\Var_{\mu}[f]}.
\end{align}
The Dirichlet form and the variational characterization of the spectral gap are indispensable to the functional-analytic approach to bounding mixing times. The approach views the Markov chain as an operator on real-valued functions over the state space, and quantifies rates of convergence to equilibrium in terms of how much a certain functional contracts in each step of the walk \cite{bobkov2006modified,saloff2006lectures}. 

The following standard comparison theorem allows one to bound the spectral gap of one walk in terms of another. We will use it in Appendix~\ref{app:sub-proof} to give an alternative proof of the spectral-gap bound for the certification walk for Slater determinant states.

\begin{theorem}[{\cite[Lemma 13.18]{levin2017markov}}]\label{thm:chain_comparison}
Let $P$ and $P'$ be reversible Markov chains on the same finite state space $\mathcal X$, with strictly positive stationary distributions $\mu$ and $\mu'$, respectively. If there exists a constant $a > 0$ such that
\begin{align}
\mathcal{E}_{P'}(f,f) \le a \,\mathcal{E}_P (f,f)
\quad \text{for all } f,
\end{align}
then
\begin{align}
\gamma(P') \le \left(\max_{x \in \mathcal X} \frac{\mu(x)}{\mu'(x)}\right) a \gamma(P).
\end{align}
\end{theorem}

Finally, let us introduce a multiaffine \emph{generating polynomial} of a distribution $\mu$ on $\{0,1\}^n$ as
\begin{align}
    Z_{\mu} (z_1,\cdots,z_n) \coloneqq \sum_{\x \in \{0,1\}^n} \mu(\x) \prod_{j=1}^n z_j^{x_j}. 
\end{align}
Because its coefficients are nonnegative, $Z_{\mu}$ has no zero when all $z_j$ are on the positive real axis.
Relaxation of Markov chains whose stationary distribution is $\mu$, however, is governed by zero-freeness of $Z$ in a complex region around the positive real axis \cite{alimohammadi2021fractional_arxiv,chen2024holant}.
This use of complex zeros is reminiscent of Lee-Yang theory, which relates partition function zeros to the analyticity of the free energy in equilibrium statistical mechanics~\cite{yang-lee1952I}. Although the generating polynomial may appear like a partition function, the relaxation bound here follows from results connecting zero-free regions to the specified down-up walks~\cite{alimohammadi2021fractional_arxiv,anari2022entropic}, rather than directly from Lee-Yang theory. The relevant zero-free property is \emph{sector stability} defined below.

\begin{definition}[Sector stability]
\label{def:sector_stable}
For $\alpha\in(0,1]$, let
\begin{align}
    \S_{\alpha} \coloneqq \{re^{i\theta} \mid r> 0, \,\theta \in (-\alpha \pi/2, \alpha \pi/2)\}.
\end{align}
be the open sector with aperture $\alpha\pi$ centered on the positive real axis. A polynomial $Z\in\C [z_1,\ldots,z_n]$ is $\S_{\alpha}$-stable if \begin{align}
    Z(z_1,\ldots,z_n)\neq 0
    \qquad
    \text{for all }
    z_1,\ldots,z_n\in\mathcal{S}_{\alpha}.
\end{align}
\end{definition}

%========================
\section{Down-up walks}\label{app:downUp}
%------------------------

%----------------------------------
\subsection{Definition and basic properties}\label{app:downUp_defn}
%----------------------------------

Down-up walks are Markov chains on fixed-size subsets $\binom{[n]}{k}$, i.e., the set of all $k$-element subsets of $[n]$.
Here we establish its basic operator properties, illustrate the connection to the well-known Glauber dynamics in statistical physics through homogenization, and show how contraction of relative entropy under the down step yields a spectral-gap bound for the full walk. These results provide the framework used in Appendix~\ref{app:main-proof} to analyze the certification walk.

\begin{definition}[Down-up walk]
Fix integers $0 \le l \le k \le n$. A $k \leftrightarrow k-l$ \emph{down-up walk} is a Markov chain $P^{\vee} : \binom{[n]}{k} \to \binom{[n]}{k}$
defined as the composition $P^{\vee} = U \circ D=DU$ of the following two steps:
\begin{itemize}
\item \textbf{Down step} $D : \binom{[n]}{k} \to \binom{[n]}{k-l}$:
given $X \in \binom{[n]}{k}$, choose a subset $Y \subset X$ of size $k-l$ uniformly at random.
\item \textbf{Up step} $U : \binom{[n]}{k-l} \to \binom{[n]}{k}$:  
given $Y$, sample $X' \supset Y$ with probability 
$\Pr(X'|Y) \propto \mu(X'),$
where $\mu$ is a target distribution supported on $\binom{[n]}{k}$.
\end{itemize}
\end{definition}
\noindent Intuitively, the down step removes $\ell$ elements uniformly at random, while the up step re-samples them according to the target distribution conditioned on the remaining subset. Explicitly, the transition matrices of the down and up steps are given by
\begin{align}
    D(X,Y) &= {k \choose l}^{\!\!-1} \text1_{Y \subset X}, \qquad
    \text1_{Y \subset X} = \begin{cases}
        1, & Y \subset X, \\
        0, & \textrm{otherwise},
    \end{cases} \\
    U(Y,X) &= \frac{\mu(X)}{\sum_{X' \supset Y} \mu(X')} \; \text1_{Y \subset X}.
\end{align}
It can be shown that the spectrum of down-up walk lies in $[0,1]$ \cite{kaufman2020high_order,alev2020improved,anari2021hardcore}, and a straightforward computation verifies that $P^{\vee}$ is reversible with respect to $\mu$. Although the down and up steps are not individually reversible with respect to $\mu$, they nonetheless satisfy a form of detailed balance due to the fact that $D$ and $U$ are adjoints of each other 
\cite{kaufman2020high_order,dikstein2024boolean}. This relationship is made precise below.

\begin{lemma}
    Suppose that $\mu$ is the stationary distribution of the down-up walk $P^{\vee}=DU$. Then, for any function $f$ on ${[n] \choose k}$ and $g$ on ${[n] \choose k-l}$, 
    \begin{align}\label{eq:downup_adjoint}
        \av{Dg,f}_{\mu} = \av{g,Uf}_{\mu D}.
    \end{align}
    In particular, it implies that $\mu$ and $\nu \coloneqq \mu D$ satisfy the following form of detailed balance:
\begin{align}\label{eq:downup_detailed_balance}
    \mu(x) D(x,y) = \nu(y) U(y,x).
    \end{align}
\end{lemma}
\noindent Before proving the lemma, let us note that the notation is consistent. Since $D$ maps distributions over ${[n] \choose k}$ to distributions over ${[n] \choose k-l}$, the composition of $D$ with $g$ must be a function on ${[n] \choose k}$:
\begin{align}
    Dg = g \circ D: {[n] \choose k} \to \R.
\end{align}
\begin{proof}
    To prove Eq.~\eqref{eq:downup_adjoint}, we simply compute and compare the left and the right sides.
    Let $X \in {[n] \choose k}$ and $Y,Y' \in {[n] \choose k-l}$
    \begin{align}
        \av{Dg,f}_{\mu} &= \sum_X \mu(X) f(X) \sum_Y D(X,Y) g(Y) \nonumber \\
        &= {k \choose l}\inv \sum_X \mu(X) f(X) \sum_{Y \subset X} g(Y)
    \end{align}
    \begin{align}
        \av{g,\!U\!f}_{\mu D} \!&= \sum_Y \left[ \sum_{X'} \mu(X') D(X',Y) \right] g(Y) \sum_X U(Y,X) f(X) \nonumber \\
        &= {k \choose l}\inv \sum_Y
        \cancel{\sum_{X'} \mu(X')} g(Y) \sum_X \left[ \frac{\mu(X)}{\cancel{\sum_{X'} \mu(X')}} \right] f(X) \nonumber \\
        &= {k \choose l}\inv \sum_X \mu(X) f(X) \sum_{Y \subset X} g(Y) = \av{Dg,f}_{\mu}
    \end{align}
In coordinates, the adjoint equation reads
\begin{align}
    \sum_{x,y} f(x) g(y) \mu(x) D(x,y)  
    &= \sum_{x,y} f(x) g(y) \nu(y) U(y,x).
\end{align}
Since this is true for any functions $f,g$, it follows that
\begin{align}
    \mu(x) D(x,y) = \nu(y) U(y,x).
\end{align}
\end{proof}

%----------------------------------
\subsection{Homogenization of Glauber dynamics as down-up walks}
%----------------------------------
It is instructive to illustrate the connection between down-up walks and statistical physics using single-site heat-bath Glauber dynamics (the Gibbs sampler), a classic method for sampling from the Boltzmann-Gibbs equilibrium distribution of classical spins $\mu(S)\propto e^{-E(S)/(k_B T)}$. Each step chooses a site uniformly at random and resamples its spin conditionally on the remaining spins. After the homogenization described below, this dynamics is a $d\leftrightarrow d-1$ down-up walk.
\begin{algorithm}[H]
\caption{Glauber dynamics (single-site update)}
\begin{algorithmic}[1]
\State Initialize a spin configuration $S \in \{\pm 1\}^d$
\Repeat
\State Choose a site $j \in [d]$ uniformly at random
\State Let $S'$ be the configuration obtained by flipping spin $j$
\State Compute $\Delta E = E(S') - E(S)$
\State Update $S \to S'$ with probability
$\frac{1}{1 + e^{\Delta E / k_B T}}$
\Until{convergence} (in distribution to equilibrium; the required runtime is model-dependent)
\end{algorithmic}
\end{algorithm}
\noindent A spin configuration $S\in\{\pm1\}^d$ can be mapped to a $d$-element subset of $[2d]$ via
\begin{align}
    S
    \;\longmapsto\;
    S^{\textrm{hom}}
    =
    \{(j,-1)\mid j\notin S\}
    \;\cup\;
    \{(j,+1)\mid j\in S\}.
\end{align}
There are two equivalent ways to view the underlying state space. One may regard each of the $d$ original sites as carrying one of two species, a ``particle'' or a ``hole'', or equivalently split each site $j$ into two sites $(j,+1)$ and $(j,-1)$, with exactly one of the two occupied. In the latter picture, the constraint that exactly one site in each pair is occupied can be thought of as a \emph{particle-hole constraint} similar to the Bogoliubov-de Gennes (BdG) construction widely used in the theory of superconductivity to describe Cooper pairing \cite{deGennes2018superconductivity,altland2010condensed}, where one doubles the number of modes by introducing particle and hole modes and thereby represents a general fermionic Gaussian state in terms of a number-conserving Gaussian state on the enlarged space (whereby the number conservation is with respect to the total number operator that includes both particles and holes).
In both cases, the doubling converts particle-number-nonconserving degrees of freedom into a number-conserving description subject to a particle-hole constraint.
In this way, Glauber dynamics can be viewed as a $d \leftrightarrow d-1$ down-up walk with the down and up steps defined as follows.
\begin{itemize}
    \item \textbf{Down step}: choose a site $j$ and remove its label (particle or hole), leaving a configuration of size $d-1$.
    \item \textbf{Up step}: reassign the label at site $j$ by choosing between the two species (equivalently, between $(j,+1)$ and $(j,-1)$), with probability proportional to the Gibbs weight of the resulting configuration.
\end{itemize}
Since only two configurations are compatible with the intermediate state, the normalization reduces to
\begin{align}
    \frac{e^{-E(S')/k_B T}}{e^{-E(S)/k_B T} + e^{-E(S')/k_B T}}
= \frac{1}{1 + e^{\Delta E / k_B T}},
\end{align}
recovering the standard Glauber update rule.

This reformulation is an instance of a standard procedure in the study of down-up walks known as \emph{homogenization} \cite{anari2019logII}, where a distribution over unconstrained configurations is embedded into a distribution over fixed-size subsets.
This technique will be crucial for identifying the certification walk with a $d \leftrightarrow d-2$ down-up walk.

Concretely, let $\mu$ be a distribution on $\{0,1\}^n$ with the generating polynomial
\begin{align}
    Z_{\mu} (z_1,\cdots,z_n) = \sum_{\x \in \{0,1\}^n} \mu(\x) \prod_{j=1}^n z_j^{x_j},
\end{align}
then its homogenization is a homogeneous  polynomial of degree $n$ in $2n$ variables,
\begin{align}
    Z_{\mu}^{hom} (z_1,\cdots,z_n,w_1,\dots,w_n) = \sum_{\x \in \{0,1\}^n} \mu(\x) \prod_{j=1}^n z_j^{x_j} w_j^{\bar x_j},
\end{align}
or in the language of subsets,
\begin{align}
    Z_{\mu}^{hom} &= \sum_{X \subseteq [n]} \mu(X) \prod_{j \in X} z_j \prod_{j \not\in X} w_j \nonumber \\
    &\eqqcolon \sum_{X \subseteq [n]}\mu(X) \, z^X w^{[n]\backslash X}.
\end{align}

%----------------------------------
\subsection{Spectral-gap bound from entropy contraction}\label{app:markov_gap_bound}
%----------------------------------

In a recent line of work \cite{anari2022entropic}, a sophisticate program in the name of \emph{entropic independence} has been put forward as a framework to obtain tight mixing time bounds for multi-level down-up walks beyond the $k\leftrightarrow k-1$ setting via establishing an entropy inequality of the form,
\begin{align}\label{eq:entropy_contraction}
    \KL{\nu D}{\mu D} \le (1-\kappa) \KL{\nu}{\mu},   
\end{align}
called entropy contraction, where 
\begin{align}
    \KL{\nu}{\mu} = \sum_{x\in\mathcal X} \nu(x) \log \left( \frac{\nu(x)}{\mu(x)} \right),
\end{align}
is the relative entropy between an arbitrary distribution $\nu$ relative to the stationary distribution $\mu$.
The spectral gap is directly lower bounded by the constant $\kappa$ in \eqref{eq:entropy_contraction} as shown in the next proposition.
\begin{proposition}[spectral-gap bound from entropy contraction]\label{prop:entropy_contraction_bounds_gap}
    Suppose there is a constant $0< \kappa <1$ such that $D$, the down operator and $\mu$ a stationary distribution for the full down-up walk $P^{\vee}=DU$,
    \begin{align}
        \KL{\nu D}{\mu D} \le (1-\kappa) \KL{\nu}{\mu}.
    \end{align}
    Then the spectral gap of the down-up walk is lower bounded by $\kappa$: 
    \begin{align}
        \gamma(P^{\vee}) \ge \kappa.
    \end{align}
\end{proposition}

\begin{remark}
    For a general walk $P$, the constant $\kappa$ also indirectly lower bounds the spectral gap through the well known modified log-Sobolev constant,
    \begin{align}
        \rho_0 = \inf_{\substack{f \\ \mathrm{Ent}_{\mu}[f] \neq 0}} \frac{\dirich{f,\log f}}{2\,\mathrm{Ent}_{\mu}[f]},
    \end{align}
    where $\mathrm{Ent}_{\mu}[f] \coloneqq \E_{\mu}[f \log f] - \E_{\mu}[f] \log (\E_{\mu} f)$. In particular \cite{bobkov2006modified,anari2021entropic},
\begin{align}\label{eq:general_entropy_contraction_bound}
        \gamma(P) \ge \rho_0(P) \ge \frac{\kappa(P)}{2}.
    \end{align}
    The last inequality is given as Lemma 16 in the preprint \cite{anari2021entropic} of \cite{anari2022entropic} after one corrects a typo; the lemma states that $\rho_0(P) \ge 2\kappa(P)$, but their proof in fact shows that $\rho_0(P) \ge \kappa(P)/2$. As a counterexample,
    \begin{align}
        P = \mqty(3/4 & 1/4 \\ 1/4 & 3/4)
    \end{align}
    has the uniform stationary distribution, $\gamma(P)=1/2$, and $\kappa(P)=3/4$.  To see the latter, define a parametrized distribution $\nu=\mqty(1/2+x & 1/2-x)$. The relative entropy as a function of $x$ is $f(x) = (1/2+x)\log(1+2x) + (1/2-x)\log(1-2x)$. After one step of $P$, the relative entropy becomes $f(x/2)$, and we can show, e.g. by plotting graphs, that $f(x/2) \le f(x)/4$, which means that
    \begin{align}
        \KL{\nu P}{\mu} \le \frac{1}{4} \KL{\nu}{\mu}.
    \end{align}
    Comparing this to $(1-\kappa(P))$, we obtain $\kappa(P)=3/4$. Since $3/4>\gamma = 1/2$, $\kappa(P)$ clearly does not lower bound the spectral gap $\gamma(P)$.

    If we naively apply the data processing inequality 
    \begin{align}
        \KL{\nu P}{\mu P} \le  \KL{\nu D}{\mu D} \le (1-\kappa)\KL{\nu}{\mu}
    \end{align}
    in combination with the general inequality \eqref{eq:general_entropy_contraction_bound} for $\kappa(P^{\vee})$, we would obtain a looser bound by a factor of 2 compared to the bound given in Proposition~\ref{prop:entropy_contraction_bounds_gap}.
\end{remark}

\begin{proof}[Proof of Proposition~\ref{prop:entropy_contraction_bounds_gap}]
The first step is to relate the relative entropy to the variance; this is possible by the well known fact that, to the lowest order, the relative entropy between two infinitesimally close distributions is quadratic in the variation. 

Consider a distribution $\nu$ that is perturbed slightly from $\mu$, i.e., $\nu = (1+\epsilon f)\mu$ where $f$ is a function with zero mean. The relative entropy reads
\begin{align}
    \KL{\nu}{\mu} &= \sum_x \mu(x) (1+\epsilon f(x)) \log (1+\epsilon f(x)) \nonumber \\
    &= \sum_x \mu(x) (1+\epsilon f(x)) \left( \epsilon f(x) - \frac{\epsilon^2 f(x)^2}{2} + O(\epsilon^3) \right) \nonumber \\
    &= \epsilon \: \cancelto{0}{\E_{\mu} [f]}
    \;\; + \frac{\epsilon^2}{2} \E_{\mu} [f^2] + O(\epsilon^3) \nonumber \\
    &= \frac{\epsilon^2}{2} \Var_{\mu}[f] + O(\epsilon^3).
\end{align}
The distribution after one down step is
\begin{align}
    (\nu D)(x) &= [(1+\epsilon f)\mu D](x) \nonumber \\
    &= (\mu D)(x) + \epsilon \sum_y f(y) \mu(y) D(y,x) \nonumber \\
    &= (\mu D)(x) + \epsilon \,[\mu D](x) \sum_y U(x,y)f(y) \nonumber \\
    &= (\mu D)(x) \left[1+\epsilon (Uf)(x) \right],
\end{align}
where we have used adjoint detailed balance~\eqref{eq:downup_detailed_balance}. Consequently, the relative entropy after one down step can be expanded in the same way provided that $\E_{\mu D}[Uf] = \E_{\mu}[f]=0$.
\begin{align}
\KL{\nu D}{\mu D} &= 
\sum_x (\mu D)(x) \left[1+\epsilon (Uf)(x) \right] \log \left[1+\epsilon (Uf) (x)\right] \nonumber \\
&= \frac{\epsilon^2}{2} \Var_{\mu D}[Uf] + O(\epsilon^3)
\end{align}
Thus, we see that entropy contraction~\eqref{eq:entropy_contraction} implies variance contraction by the same amount:
\begin{align}\label{eq:variance_contraction}
    \Var_{\mu D}[Uf] \le (1-\kappa) \Var_{\mu}[f].
\end{align}
Now, the adjoint property of $D$ and $U$ lets us write
\begin{align}
    \Var_{\mu D} (Uf) 
    &= \sum_x \, [\mu D](x) [(Uf)(x)]^2 \nonumber \\
    &= \av{Uf,Uf}_{\mu D} = \av{f,DUf}_{\mu} = \av{f,Pf}_{\mu} \nonumber \\
    &= \Var_{\mu}[f] - \dirich{f,f}.
\end{align}
Rearranging the terms and applying \eqref{eq:variance_contraction},
    \begin{align}
        \dirich{f,f} 
        &= \Var_{\mu}[f] - \Var_{\mu D}[Uf] \nonumber \\
        &\ge \Var_{\mu}[f] - (1-\kappa) \Var_{\mu}[f] \nonumber \\
        &= \kappa \Var_{\mu}[f]
    \end{align}
Since this is true for any $f$, it must be true for $f$ that minimizes the ratio $\mathcal E(f,f)/\Var_{\mu}[f]$. In other words,
\begin{align}
    \gamma = \inf_{\substack{f\\ \E_{\mu}[f]=0 \\ \Var_{\mu}[f] \neq 0}} \frac{\dirich{f,f}}{\Var_{\mu}[f]} \ge \kappa.
\end{align}
\end{proof}

%==========================
\section{Proof of Theorem~\ref{thm:main}}\label{app:main-proof}
%==========================

In this appendix, we provide a proof of the main theorem, restated here for the reader's convenience.

\SpectralGap*

There are mainly two stages to the proof. We first show that the generating polynomial associated with the computational-basis distribution in Eq.~\eqref{eq:PfaffianProb},
\begin{align}
    \pi(\x)
    =
    \abs{\av{\x|\psi}}^2
    =
    \frac{1}{2^d}
    \abs{\Pf\bigl(\Gamma+D_{\x}J\bigr)},
\end{align}
is $\mathcal{S}_1$-stable (Definition~\ref{def:sector_stable}), and hence zero-free when all variables lie in the open right half-plane. We then show that, after homogenization, the certification walk $P$ of Eq.~\eqref{eq:certification_walk} is precisely a $d\leftrightarrow d-2$ down-up walk with stationary distribution $\pi$. This allows us to apply the machinery of Refs.~\cite{alimohammadi2021fractional_arxiv,anari2022entropic}: sector stability implies entropy contraction of the associated down-up walk, which in turn yields the desired spectral-gap bound via Proposition~\ref{prop:entropy_contraction_bounds_gap}, proved in Appendix~\ref{app:downUp}.

For convenience, let us absorb the sign of the Pfaffian in the measurement probability, which depends on both the parity of the Gaussian state and the Majorana ordering, as discussed in Appendix~\ref{app:grassmann}, into a factor $\zeta\in\{\pm1\}$, so that
\begin{align}
    \pi(\x)
    =
    \frac{1}{2^d}
    \abs{\Pf\bigl(\Gamma+D_{\x}J\bigr)}
    = \frac{\zeta}{2^d}
    \Pf\bigl(\Gamma+D_{\x}J\bigr).
\end{align} 
Once the convention is fixed, $\zeta$ is fixed globally across all outcomes $\x$.
The next proposition shows that the generating polynomial $Z_{\pi}$ is effectively a Pfaffian in disguise.

\begin{proposition}\label{prop:CayleyTransform}
Let $\Upsilon(\vb s)$ be an extension of the distribution $\pi(\s) = (\zeta/2^d) \,\Pf (\Gamma + D_{\s} J)$, thought of as a function of $\s$, to the complex domain. That is, $s_j$ can now take an arbitrary complex value. The generating polynomial of the certification walk can be written in the form
\begin{align}\label{eq:CayleyIdentity}
     Z_{\pi}(\z) =  \Upsilon(\vb w) \cdot \prod_{j=1}^d (1+z_j), \qquad w_j = \frac{1-z_j}{1+z_j}. 
\end{align}
Consequently, $Z_{\pi}$ is zero-free in the open right half-plane for all $z_j$ if and only if $\Upsilon$ is zero-free in the open unit disk for all $w_j$.
\end{proposition}

Before we give a proof of Proposition~\ref{prop:CayleyTransform}, let us state a simple lemma for ease of invoking it later.
\begin{lemma}\label{lem:DOmGrassmann}
For $\s \in \C^d$, $D_{\s} = \mathrm{diag}(\s) \oplus \mathrm{diag}(\s)$, and $\theta=(\theta_1, \dots, \theta_{2d})^{\mathsf{T}}$ a column vector of $2d$ Grassmann variables,
\begin{align}
    e^{\frac{1}{2} \theta^{\mathsf{T}} D_{\s} J \,\theta} = \prod_{j=1}^d (1 + s_j \theta_j \theta_{d+j}).
\end{align}
\end{lemma}

\begin{proof}[Proof of Lemma~\ref{lem:DOmGrassmann}]
Since $D_{\s} J = \mqty(0&\textrm{diag}(\s) \\ -\textrm{diag}(\s)&0)$,
\begin{align}
    \theta^\T D_{\s} J \, \theta 
    = \sum_{j=1}^d s_j \theta_j \theta_{d+j} -  \sum_{j=1}^d s_j \theta_{d+j} \theta_j 
    = 2\sum_{j=1}^d s_j \theta_j \theta_{d+j}.
\end{align}
All terms in the sum are even Grassmann elements, so they are mutually commuting and the exponential factorizes. Each exponential truncates after the linear term due to Eq.~\eqref{eq:GrassmannExp}, proving the lemma.
\end{proof}

\begin{proof}[Proof of Proposition~\ref{prop:CayleyTransform}]
Using Lemma~\ref{lem:DOmGrassmann} and the fact that $s_j = (-1)^{x_j}$, 
\begin{align}
Z_{\pi}(\z) 
    &= \frac{\zeta}{2^d} \sum_{\s\in \{\pm 1\}^d} \int D\theta \; e^{\frac{1}{2} \theta^\T (\Gamma + D_{\s} J) \theta}  
    \prod_{j=1}^d z_j^{x_j} \nonumber \\
    &= \frac{\zeta}{2^d} \sum_{\s} \int D\theta \; e^{\frac{1}{2} \theta^\T \Gamma \theta}  
    \prod_{j=1}^d z_j^{x_j} (1 + (-1)^{x_j} \theta_j \theta_{d+j}).
\end{align}
Exchanging the sum and the product,
\begin{align}
Z_{\pi}(\z) 
    &= \frac{\zeta}{2^d} \int D\theta \; e^{\frac{1}{2} \theta^\T \Gamma \theta}  
    \prod_{j=1}^d \sum_{x_j=0}^1 z_j^{x_j}(1 + (-1)^{x_j} \theta_j \theta_{d+j}) \nonumber\\
    &= \frac{\zeta}{2^d}  \int D\theta \; e^{\frac{1}{2} \theta^\T \Gamma \theta}  
    \prod_{j=1}^d (1 + z_j + (1-z_j) \theta_j \theta_{d+j})  \nonumber\\
    &= \frac{\zeta}{2^d} \int D\theta \; e^{\frac{1}{2} \theta^\T \Gamma \theta} 
    \prod_{j=1}^d  \left(1 + w_j \theta_j \theta_{d+j}\right) 
    \cdot \prod_{j=1}^d \left(1 + z_j\right).
\end{align}
Finally, using Lemma~\ref{lem:DOmGrassmann} again to roll the product back into an exponential $\exp\left(\theta^{\T} D_{\vb w} J \, \theta/2 \right)$ gives the result,
\begin{align}
Z_{\pi}(\z) =  \Upsilon(\vb w) \cdot \prod_{j=1}^d (1+z_j).
\end{align}
The Cayley transform $z\mapsto(1-z)/(1+z)$ bijectively maps the open right half-plane to the open unit disk, inside which the denominator $1+z_{j}$ is never zero, yielding the final claim.
\end{proof}

\begin{lemma}{\label{lem:ZeroFreeness}}
    The polynomial $\Upsilon(\s)$ is zero-free in the open unit polydisk, i.e., whenever $|s_j|<1$ for all $j$.
\end{lemma}

\begin{proof}
    $\Pf \,A = 0 \iff \det \,A =0 \iff A $ has a nontrivial kernel. So let us proceed to prove the lemma by assuming there is a nonzero vector $u \in \ker (\Gamma + D_{\s} J)$. In other words,
    \begin{align}
        \Gamma u = - D_{\s} J \, u.
    \end{align}
    Now take the 2-norm of both sides. For pure Gaussian states, $\Gamma \in \SO(2d)$, i.e., $\Gamma^\T \Gamma = -\Gamma^2 = \id$. Thus, the left hand side is just $\norm{u}$. Similarly, $J \in \SO(2d)$ and $\norm{J u} = \norm{u}$. We therefore have that
    \begin{align}
        \norm{u} \le \max_j{\abs{s_j}} \norm{u}.
    \end{align}
    However, since every $s_j$ lies in the open unit disk, the inequality $\abs{s_j}<1$ is strict and we arrive at the contradiction that $\norm{u} < \norm{u}$.
\end{proof}

Proposition~\ref{prop:CayleyTransform} and Lemma~\ref{lem:ZeroFreeness} together show that the generating polynomial of the certification walk is $\S_1$-stable. Next, we identify the homogenization of the certification walk with a down-up walk.

\begin{lemma}
    The homogenization of the certification walk $P$ in
    Eq.~\eqref{eq:certification_walk} is a
    $d\leftrightarrow d-2$ down-up walk.
\end{lemma}
\begin{proof}
Recall that the certification walk is a walk on the graph $G=(V,E)$, where $V=\{\x \mid \pi(\x)>0\}$ is the support of the distribution $\pi$ over $\{0,1\}^d$, and two vertices are connected iff they differ on exactly two modes, $|\x\oplus\y|=2$. In particular, the walk always preserves the parity of the Hamming weight.

Homogenization corresponds to replacing each mode $j\in[d]$
by a pair of sites $(j,0)$ and $(j,1)$:
\begin{align}
    X \to X^{hom} = \{(j,0)\mid x_j=0\} \cup \{(j,1)\mid x_j=1\}.
\end{align}
Thus $X^{hom} \in\binom{[2d]}{d}$ contains exactly one occupied site from each pair $\{(j,0),(j,1)\}$.
Equivalently, one may regard the original configuration as $d$ sites, each carrying one of two labels, $0$ or $1$.

Consider the $d\leftrightarrow d-2$ down-up walk with stationary distribution $\pi$. 
\begin{itemize}
    \item \textbf{Down step:} starting from $\x$, choose an unordered pair of modes $k=\{k_1,k_2\}\in\binom{[d]}{2}$ uniformly and remove their two labels $(x_{k_1},x_{k_2})$, leaving a partial configuration $\z^{(k)}$ on the remaining $d-2$ modes. 
    \item \textbf{Up step:} assign new labels $(y_{k_1},y_{k_2})$ to the two modes with probability proportional to the stationary weight of the resulting configuration $\y$. Explicitly, if we denote by $Z^{(k)}$ the subset corresponding to $\z^{(k)}$, then
    \begin{align}
    \Pr\bigl(Y|Z^{(k)} \bigr)
    =
    \frac{\pi(Y)}
    {\sum_{X' \supset Z^{(k)}} \pi(X')}.
\end{align}
\end{itemize}
\noindent Since $\pi$ is supported on a fixed-parity sector, only two of the
four possible assignments of $(y_{k_1},y_{k_2})$ can have nonzero weight: if $\y$ agrees with $\x$ on $[d]\backslash k$, then either $\y=\x$ or $|\x \oplus \y|=2$.
Thus, 
\begin{align}\label{eq:down_up_offdiag}
    \mel{\x}{P^{\vee}}{\y}
    =
    \binom{d}{2}\inv
    \frac{\pi(\y)}
    {\pi(\x)+\pi(\y)}, \quad (\x,\y)\in E.
\end{align}
The diagonal transition probability is obtained by summing the remaining probability over all choices of $k$ for which the up step returns to $\x$. Eq.~\eqref{eq:down_up_offdiag} is precisely the transition probability of the certification walk in
Eq.~\eqref{eq:certification_walk} for every pair of distinct configurations in the support of $\pi$. Since both $P$ and $P^{\vee}$ are stochastic, their diagonal entries also agree. We therefore conclude that
\begin{align}
    P^{\vee}=P,
\end{align}
which proves that the homogenized certification walk is the
$d\leftrightarrow d-2$ down-up walk with stationary distribution $\pi$.
\end{proof}

We now invoke two general results to connect sector stability to entropy contraction of multi-level down-up walks.

\begin{theorem}[Lemma 71 of the preprint \cite{alimohammadi2021fractional_arxiv}]\label{thm:FLC-hom}
    Let $\mu$ be a distribution generated by a multiaffine $\S_{\alpha}$-stable polynomial $Z_{\mu}$. Its homogenization $Z_{\mu}^{hom}$ is $\alpha/2$-fractionally log-concave.
\end{theorem}

\begin{theorem}[Theorem 5 of \cite{anari2022entropic}]\label{thm:FLC-up-down}
    Suppose that $\mu$ is a distribution on $\binom{[n]}{k}$ whose generating polynomial is $\alpha$-fractionally log-concave for $0<\alpha\le1$. Let $l \le k-\lceil 1/\alpha \rceil$, and $D_{k\to l}$ the corresponding down operator. For all probability distribution $\nu$ on the support of $\mu$,
    $$\KL{\nu D}{\mu D} \le (1-\kappa) \KL{\nu}{\mu},$$
    where $\kappa$ is a function of $k,l$, and $\alpha$. In particular, when $1/\alpha$ is an integer, 
    \begin{align}
        \kappa = \left. {k-l \choose 1/\alpha} \middle/ {k \choose 1/\alpha} \right..
    \end{align}
\end{theorem}
\noindent Specializing to the $k\leftrightarrow k-2$ down-up walk, the theorem says that if $Z_{\mu}^{hom}$ is 1/2-fractionally log-concave, then the down walk contracts entropy with $\kappa = {k \choose 2}^{-1}$. 
 
Theorem~\ref{thm:FLC-hom} implies that the homogenized $Z_{\pi}^{hom}$ of the certification walk is 1/2-fractionally log-concave.
By Theorem~\ref{thm:FLC-up-down}, the down step of the corresponding $d\leftrightarrow d-2$ down-up walks with $\pi$ as the stationary distribution contracts the relative entropy with the particular set of constants, $k=d$, $l=d-2$ and $\alpha=1/2$, so that
\begin{align}
    \kappa = \binom{d}{2}^{\!\!-1} = \frac{2}{d(d-1)}.
\end{align}
As shown in Proposition~\ref{prop:entropy_contraction_bounds_gap}, $\kappa$ lower bounds the spectral gap $\gamma(L)$ of the certification protocol.  Hence, we conclude the proof of Theorem~\ref{thm:main}.

%==========================
\section{Shorter proof of Theorem~\ref{thm:main} for Slater determinants}\label{app:sub-proof} 
%==========================

Since Slater determinants are pure Gaussian states, and total particle number can be certified from the same computational-basis measurement used to certify parity, Algorithm~\ref{algo} certifies target Slater determinants with the same sample-complexity bound as in Theorem~\ref{thm:main}. Moreover, the pair-product states in Proposition~\ref{prop:tight} saturate the bound in Theorem~\ref{thm:main} with matching constant whether they are number-conserving or not, showing the sharpness of the bound for this subclass of states as well.

For Slater determinants, however, the proof of Theorem~\ref{thm:main} can be considerably shortened by leveraging two special structures.
First, the number-conserving certification walk can be compared directly with the $n\leftrightarrow n-1$ down-up walk without the need for homogenization.
Second, the spectral gap of this down-up walk is known for the determinantal stationary distribution, Eq.~\eqref{eq:DeterminantProb}; the distribution belongs to the class of \emph{determinantal point processes} (DPPs), which have been extensively studied in stochastic processes and machine learning, see, for example, Ref.~\cite{lyons2003determinantal,kulesza2012determinantal}. 

We begin with the comparison between the two walks.

\begin{lemma}[Dirichlet-form comparison]
\label{lem:NC-dirichlet-compare} Let $P$ be the number-conserving certification walk, and $P^{\vee}$ be the $n\leftrightarrow n-1$ down-up walk with the same stationary distribution as $P$. Then for every real-valued $f$, we have that
\begin{align}
    \mathcal E_P(f,f)\ \ge\ \dfrac{n}{\binom{d}{2}}\,\mathcal E_{P^\vee}(f,f).
\end{align}
\end{lemma}
\begin{proof}
Utilizing the expression for the Dirichlet form in Eq.~\eqref{eq:dirichletDiag},
\begin{align}
\mathcal E_P (f,f) =  \frac{1}{2} \sum_{x,y} (f(x)-f(y))^2 \pi(x) P(x,y),
\end{align}
it suffices to compare the stationary edge weights $\pi(\x)P(\x,\y)$ and $\pi(\x)P^{\vee}(\x,\y)$ for each off-diagonal pair $\x \neq \y$.

For the number-conserving certification walk, a transition between distinct configurations is possible when $|\x\oplus\y|=2$. In terms of the
corresponding subsets, this means $Y=(X\setminus\{j\})\cup\{k\}$ for $j\in X$, $k\notin X$.
For the $n\leftrightarrow n-1$ down-up walk, the same transition is
obtained by first removing $j$ from $X$, giving
$Z=X\setminus\{j\}=X\cap Y$,
and then adding $k$ to return to an $n$-element set.
The two walks differ only in the normalization of this two-step transition. For the down-up walk, the down step selects $j$ with probability $1/n$, while the up step selects an $n$-element set containing $Z$ with probability proportional to its stationary weight. Hence
\begin{align}
    P^{\vee}(X,Y) &= \frac1n \frac{\pi(Y)}{\sum_{X' \supset Z} \pi(X')}
\end{align}
In contrast, the certification walk has
\begin{align}
    P(X,Y) &= \binom{d}{2}\inv \frac{\pi(Y)}{\pi(X)+\pi(Y)}.
\end{align}
Since $\sum_{X'\supset Z}\pi(X')
    \geq \pi(X)+\pi(Y),$ we obtain
\begin{align} 
\pi(X)P(X,Y) &=
    \binom{d}{2}\inv \frac{\pi(X)\pi(Y)}{\pi(X)+\pi(Y)} \nonumber \\
    &\ge \frac{n}{\binom{d}{2}} \cdot\frac1n\frac{\pi(X)\pi(Y)}{\sum_{X' \supset Z} \pi(X')} \\
    &= \frac{n}{\binom{d}{2}}\,\pi(X)P^\vee(X,Y).
\end{align}
Summing over all pairs then gives the claimed Dirichlet-form comparison.
\end{proof}

Since $P$ and $P^{\vee}$ here share the same stationary distribution, Lemma~\ref{lem:NC-dirichlet-compare} together with the standard comparison theorem, Theorem~\ref{thm:chain_comparison}, imply the spectral-gap lower bound,
\begin{align}\label{eq:NC-gap-bound}
    \gamma(P) \ge \frac{n}{\binom{d}{2}} \, \gamma(P^{\vee}),
\end{align}
for the certification walk.

Now we formally define a DPP.
A probability distribution over $X \subseteq {[d] \choose k}$ is a DPP if there exists a positive semidefinite matrix $L$ such that $\pi(X) \propto \det(L_{X,X})$ .
Here, $L$ may in general be a complex Hermitian positive definite matrix. For fixed cardinality, this is often referred to as a \emph{$k$-DPP}.
For a Slater determinant, Eq.~\eqref{eq:DeterminantProb} reads
\begin{align}
    \pi(\x) = \abs{\av{\x|\psi}}^2 = \det(K_{X,X}),
\end{align}
where $K$ is a rank-$n$ projection operator and hence positive semidefinite. $\pi$ is therefore an $n$-DPP with kernel $K$.

We now invoke a general chain of implications from the literature that establishes a spectral-gap bound for $n\leftrightarrow n-1$ down-up walks whose stationary distributions are DPPs.

\begin{theorem}
\label{thm:SLC}
The distribution $\pi(X)=\det (K_{X,X})$, $X\in\binom{[d]}{n}$,
is homogeneous strongly log-concave.
\end{theorem}

\begin{proof}
This follows from the fact that DPPs are strongly Rayleigh \cite{borcea2009negative}. The generating
polynomial $Z_{\pi}$ is therefore stable. Moreover, since $Z_{\pi}$ is homogeneous with nonnegative coefficients, the polynomial is Lorentzian
\cite[Proposition~2.2]{branden2020lorentzian}. Homogeneous Lorentzian polynomials are strongly log-concave \cite[Theorem~2.30]{branden2020lorentzian}.
\end{proof}

\begin{theorem}\label{thm:alov}
If $\mu$ is a homogeneous strongly log-concave distribution on $\binom{[d]}{n}$, then the corresponding down-up walk has spectral gap $\gamma(P^\vee)\geq 1/n$.
\end{theorem}

\noindent Theorem~\ref{thm:alov} is Theorem 1.1 of Ref.~\cite{anari2019logII}.
Taken together, Eq.~\eqref{eq:NC-gap-bound}, Theorem~\ref{thm:SLC}, and Theorem~\ref{thm:alov} imply a spectral-gap bound for Slater determinants, 
\begin{align}
    \gamma(P) \ge \frac{n}{\binom{d}{2}} \, \gamma(P^{\vee}) \ge \frac{n}{\binom{d}{2}} \cdot \frac1n = \frac{2}{d(d-1)},
\end{align}
matching the general bound in Theorem~\ref{thm:main}.

%==========================
\section{Numerical methods}\label{app:numerics}
%==========================

Every point in Fig.~\ref{fig:gap} is obtained by exact computation: the
Born distribution $\pi$ of the target is evaluated on its full support, the certification walk \eqref{eq:certification_walk}
is assembled as an explicit matrix on that support, and its spectral gap
is extracted from the spectrum of that matrix. No sampling or Monte Carlo
estimation enters at any stage, so the only errors are those of
floating-point arithmetic. The cost of exact enumeration grows
exponentially in $d$,
since the support has up to $2^{d-1}$ elements; we take $d\le14$, where
that is $8192$.

\subsection{Born distributions}

For a target specified by its covariance matrix $\Gamma$ we generate
$\pi$ by the chain rule for conditional probabilities, measuring modes
$1,2,\dots,d$ in turn. Mode $k$ yields outcome $x_k$ with conditional
probability $p=(1+s\,\Gamma_{k,d+k})/2$, $s=(-1)^{x_k}$, after which the
remaining $d-1$ modes are again Gaussian with (see, e.g. \cite{bravyi2012disorder})
\begin{align}\label{eq:numerics-update}
\Gamma'_{jl}=\Gamma_{jl}
-\frac{\Gamma_{jm_{1}}\Gamma_{lm_{2}}-\Gamma_{jm_{2}}\Gamma_{lm_{1}}}
       {s+\Gamma_{m_{1}m_{2}}},
\end{align}
where $m_{1}=k$, $m_{2}=d+k$, and $j,l$ run over the remaining $2d-2$
Majorana indices. Recursing over
both outcomes enumerates the support
$V$ together with the exact probabilities $\pi(\x)$ in $O(\abs{V}d^{2})$
time. This is preferable to evaluating the Pfaffian formula
\eqref{eq:PfaffianProb} independently at each of the $2^{d}$ bitstrings,
and the two were checked against each other (Sec.~\ref{app:numerics-checks}).

A branch is discarded when its \emph{conditional} probability falls below
$10^{-12}$. Pruning on the conditional rather than the cumulative
probability retains genuinely small $\pi(\x)$ that arise as a product of
many moderate factors, while removing branches that are nonzero only at
the level of machine precision, in particular the leakage into the wrong
parity sector. A final filter projects onto the dominant parity sector and
verifies that the discarded probability is below $10^{-10}$; the retained
distribution is checked to be normalized to within $10^{-9}$.

For the spin chains of Sec.~\ref{app:numerics-families} the ground state is
computed as a $2^{d}$-dimensional state vector by exact diagonalization,
and $\pi$ is read off from its amplitudes with a threshold of $10^{-15}$.
Since these
states are Gaussian, this provides an independent route to $\pi$ that does
not use the covariance formalism at all.

\subsection{Spectral gap of the certification walk}

The transition matrix $P$ of Eq.~\eqref{eq:certification_walk} is
reversible with respect to $\pi$, so we work with the symmetric kernel
$A = S^{\,\hspace{-1pt}\frac{1}{2}}  P S^{\,\text{-}\hspace{-1pt}\frac{1}{2}}$, $S=\mathrm{diag}(\pi)$, with off-diagonal entries
\begin{align}\label{eq:symmetrized-kernel}
A_{\x\y}=\binom{d}{2}^{\!\!-1}\frac{\sqrt{\pi(\x)\pi(\y)}}{\pi(\x)+\pi(\y)},
\qquad \abs{\x\oplus\y}=2 ,
\end{align}
and $A_{\x\x}=\mel{\x}{P}{\x}$, so that $A\sqrt{\pi}=\sqrt{\pi}$. Since
$L$ and $P$ are isospectral (Appendix~\ref{app:verify-op-is-markov}),
$\gamma(L)=1-\lambda_{2}(A)$.
The similarity transformation leaves the spectrum unchanged while making it
real by construction, so $\lambda_{2}$ can be extracted with a symmetric
eigensolver, and by Lanczos rather than Arnoldi iteration in the sparse
regime. Accuracy is the main gain: eigenvalues of a symmetric matrix are
perfectly conditioned, whereas for the non-symmetric $P$ the error in
$\lambda_{2}$ would be amplified by the eigenvector condition number, of
order $\sqrt{\pi_{\max}/\pi_{\min}}\approx2\times10^{3}$ at $d=14$.

For $\abs{V}\le1500$ the full spectrum of $A$ is computed by dense
diagonalization, and above that cutoff its two largest eigenvalues by
Lanczos iteration started from the Perron vector $\sqrt{\pi}$. Every run
asserts $\abs{\lambda_{1}-1}<10^{-8}$.

\subsection{State families and model Hamiltonians}
\label{app:numerics-families}

All families use $d\in\{4,6,8,10,12,14\}$; the $d=2m$ product families
require even $d$, and half filling requires even $d$ as well.

\paragraph*{Haar-random Gaussian states.}
$\Gamma=OJO^{\T}$ with $O$ drawn from the Haar measure on $\mathrm{O}(2d)$
(SciPy \texttt{ortho\_group}) and $J$ as in Eq.~\eqref{eq:symplectic_form}.
Since $\Pf(OJO^{\T})=\det(O)\,\Pf (J)$ by \eqref{eq:pfaffian-basis-change}
and the sign of $\Pf(\Gamma)$ is the
fermionic parity, the half of the samples drawn from the $\det(O)=-1$
component of $\mathrm{O}(2d)$ are odd-parity states: both parity sectors
are sampled with equal weight. The support is the whole sector,
$\abs{V}=2^{d-1}$.

\paragraph*{Haar-random Slater determinants.}
The orbital matrix $U\in\C^{d\times n}$ at half filling $n=d/2$ is the
first $n$ columns of a Haar-random unitary in $\U(d)$ (SciPy
\texttt{unitary\_group}), and the Born distribution is
Eq.~\eqref{eq:DeterminantProb} with the 1-RDM $K=UU\dgg$.
The support is the fixed-particle-number sector,
$\abs{V}=\binom{d}{d/2}$.

\paragraph*{XX chain.}
The number-conserving hopping model with open boundaries,
\begin{align}\label{eq:xx-chain}
H_{\mathrm{XX}}=-\sum_{j=1}^{d-1}\bigl(a\dgg_{j}a_{j+1}+\mathrm{h.c.}\bigr),
\end{align}
whose ground state at filling $n=d/2$ is the Slater determinant built from
the $n$ lowest eigenvectors of the $d\times d$ hopping matrix.

\paragraph*{XY and transverse-field Ising chains.}
The open spin chain
\begin{align}\label{eq:xy-chain}
H_{\mathrm{XY}}=-\sum_{j=1}^{d-1}\Bigl[
 \tfrac{1+\gamma_{xy}}{2}X_{j}X_{j+1}
+\tfrac{1-\gamma_{xy}}{2}Y_{j}Y_{j+1}\Bigr]
-g\sum_{j=1}^{d}Z_{j},
\end{align}
which the Jordan--Wigner transformation maps to a quadratic Majorana
Hamiltonian, so its eigenstates are Gaussian. Figure~\ref{fig:gap} shows
the critical Ising point $(\gamma_{xy},g)=(1,1)$ and the anisotropic case
$(\gamma_{xy},g)=(1/2,1)$. Ground states are obtained by Lanczos
iteration in the full $2^{d}$-dimensional space and are nondegenerate for
the plotted parameters (the many-body gap $E_{1}-E_{0}$ stays above $0.11$
up to $d=14$); each is verified to have definite fermionic parity before
its Born distribution is used.

\paragraph*{Worst-case families.}
The two families of Proposition~\ref{prop:tight} factorize over the $m=d/2$
mode pairs, so $\pi$ is assembled as a product over pairs rather than
through the recursion above: $\bigotimes_{r}(\cos\theta\ket{00}
+\sin\theta\ket{11})$ at $\theta\in\{\pi/4,\pi/6,\pi/12\}$, and
$\bigotimes_{r}(\ket{01}+\ket{10})/\sqrt2$. Here $\abs{V}=2^{d/2}$.

\subsection{Sampling}

Each random ensemble contributes $60$ independent instances at every $d$,
$720$ in total, with no instance discarded. Including the
deterministic chain and worst-case families, Fig.~\ref{fig:gap} summarizes
$774$ exactly computed gaps. Table~\ref{tab:gap-summary} lists the
resulting medians and minima.

\begin{table}[t]
\caption{Spectral gap $\gamma(L)$ over $60$ random instances per $d$ and
per ensemble, compared with the bound \eqref{eq:gapbound}. Every instance
lies above the bound.}
\label{tab:gap-summary}
\footnotesize
\begin{tabular}{@{}lccccc@{}}
\toprule
 & \multicolumn{2}{c}{random fGS} & \multicolumn{2}{c}{random Slater} & \\
\cmidrule(lr){2-3}\cmidrule(lr){4-5}
$d$ & median & min & median & min & $2/d(d-1)$ \\
\midrule
 4 & 0.3015 & 0.2384 & 0.2333 & 0.1795 & 0.1667 \\
 6 & 0.1655 & 0.1219 & 0.1215 & 0.0873 & 0.0667 \\
 8 & 0.1251 & 0.1010 & 0.0780 & 0.0557 & 0.0357 \\
10 & 0.0970 & 0.0792 & 0.0587 & 0.0415 & 0.0222 \\
12 & 0.0836 & 0.0717 & 0.0458 & 0.0323 & 0.0152 \\
14 & 0.0717 & 0.0649 & 0.0395 & 0.0322 & 0.0110 \\
\bottomrule
\end{tabular}
\end{table}

Least-squares fits of $\gamma\propto d^{-\alpha}$ over $d\ge6$ give
$\alpha=0.99$ for the median of the Haar-random Gaussian states and
$\alpha=1.34$ for the median of the random Slater determinants, against
$\alpha=1.82$ for the critical Ising chain, $1.88$ for the XY chain,
$1.79$ for the XX chain, and the exact $\gamma=2/d(d-1)$ of the worst-case
families.

\subsection{Consistency checks}\label{app:numerics-checks}

The implementation is cross-validated at small $d$: the Born distributions
produced by the covariance recursion agree with exact diagonalization of
the corresponding Jordan--Wigner Hamiltonian and with the Pfaffian formula
\eqref{eq:PfaffianProb} evaluated bitstring by bitstring; building the
verification operator $\Omega$ of Algorithm~\ref{algo} explicitly confirms
$\Omega\ket{\psi}=\ket{\psi}$ and $\gamma(\Omega)=\gamma(L)/2$; and the
worst-case families reproduce the closed-form spectrum of
Proposition~\ref{prop:tight}. All agree to $10^{-8}$ or better. Over all
$774$ instances the smallest value of
$\gamma(L)\big/\bigl[2/d(d-1)\bigr]$ is $1$ to within $10^{-13}$, so no
instance violates Theorem~\ref{thm:main}. The computation is implemented
in Python with NumPy and SciPy.

\end{document}